\documentclass[reqno]{amsart}

\usepackage[bbgreekl]{mathbbol}
\usepackage{amsthm,amsmath,amsfonts,amssymb,graphicx,latexsym,bm,nicefrac,dsfont,bbm,esint,mathrsfs} 
\usepackage[scr=stixfancy]{mathalpha}
\usepackage{comment}

\DeclareSymbolFontAlphabet{\mathbbm}{bbold}
\DeclareSymbolFontAlphabet{\mathbb}{AMSb}

\usepackage{multirow,float,xcolor}
\usepackage{url}
\usepackage{hyperref}
\usepackage[hyphenbreaks]{breakurl}
\hypersetup{
  colorlinks   = true, 
  urlcolor     = blue, 
  linkcolor    = blue, 
  citecolor    = blue, 
  breaklinks   = true
}

\usepackage{csquotes}
\usepackage[style=alphabetic,sorting=debug,maxbibnames=99]{biblatex} 

\newtheorem{definition}{Definition}[section]
\newtheorem{theorem}[definition]{Theorem}
\newtheorem*{theorem*}{Theorem}
\newtheorem{lemma}[definition]{Lemma}
\newtheorem{proposition}[definition]{Proposition}
\newtheorem{corollary}[definition]{Corollary}

\newtheorem{algorithm}[definition]{Algorithm}
\newtheorem{model}[definition]{Model} 

\newtheorem{discretization}[definition]{Discretization}

\renewcommand{\theequation}{\arabic{section}.\arabic{equation}}

\numberwithin{equation}{section}
\numberwithin{table}{section}
\numberwithin{figure}{section}

\allowdisplaybreaks

\newcommand{\C}{\mathbb{C}}
\newcommand{\N}{\mathbb{N}}

\newcommand{\R}{\mathbb{R}}
\newcommand{\Z}{\mathbb{Z}}

\newcommand{\B}{\mathcal{B}}

\newcommand{\eps}{\varepsilon}

\newcommand{\E}{\mathbb{E}}
\renewcommand{\Pr}{P}

\renewcommand{\d}{ \mathrm{d}}	

\bmdefine\h{h}
\bmdefine\P{P}     
\bmdefine\r{r}
\bmdefine\u{u} 
\bmdefine\v{v} 
\bmdefine\x{x}   
\bmdefine\y{y}
\bmdefine\z{z}
\bmdefine\w{w} 
\bmdefine\balpha{\alpha} 
\bmdefine\kap{\kappa}
\bmdefine\btheta{\theta} 
\bmdefine\bxi{\xi}  
\bmdefine\bom{\omega}

\newcommand{\Pt} {\P(\btheta)}
\bmdefine\bzeta{\zeta}
\bmdefine\etab{\eta}
\bmdefine\bzeta{\zeta}
\bmdefine\bphi{\varphi}
\bmdefine\Z{\mathcal Z}
\bmdefine\bZ{Z}

\newcommand{\ebb}{\mathbbm{e}}

\newcommand{\etabb}{\bbeta}

\newcommand{\Pbb}{\mathbbm{P}}

\newcommand{\wrt}{w.r.t.\ }

\newcommand{\bdot}{\,\bm{\cdot}\,}

\newcommand{\ind}{\mathds{1}}

\newcommand{\st}{\sigma_{\mathrm{t}}}
\newcommand{\su}{\sigma_{\mathrm{u}}}
\newcommand{\sx}{\sigma_{\mathrm{x\;\!}}}
\newcommand{\sz}{\sigma_{\mathrm{z}}}

\newenvironment{psmallmatrix}
  {\left(\begin{smallmatrix}}
  {\end{smallmatrix}\right)}

\newcommand{\wnr}{\xi}
\newcommand{\wn}{\bm{\xi}}

\newcommand{\flow}{\bphi}

\makeatletter

\newcommand*\wthelper[2]{%
        \hbox{\dimen@\accentfontxheight#1%
                \accentfontxheight#11.2\dimen@
                $\m@th#1\widetilde{#2}$%
                \accentfontxheight#1\dimen@
        }%
}

\newcommand*\accentfontxheight[1]{%
        \fontdimen5\ifx#1\displaystyle
                \textfont
        \else\ifx#1\textstyle
                \textfont
        \else\ifx#1\scriptstyle
                \scriptfont
        \else
                \scriptscriptfont
        \fi\fi\fi3
}
\makeatother

\begin{document}

\title[]{
Random field reconstruction of\\ 
inhomogeneous turbulence \\
Part II: Numerical approximation and simulation}
\author[Antoni et al.]{Markus Antoni$^{1}$}
\author[]{Quinten K\"urpick$^{1}$}
\author[]{Felix Lindner$^{1}$}
\author[]{Nicole Marheineke$^{2}$}
\author[]{Raimund Wegener$^{3}$}

\date{\today\\
$^1$ Universit\"at Kassel, Institut f\"ur Mathematik, Heinrich-Plett-Str.~40, 
D-34132 Kassel, Germany\\
$^2$ Universit\"at Trier, Arbeitsgruppe Modellierung und Numerik, Universit\"atsring 15, D-54296 Trier, Germany\\
$^3$ Petrusstr. 1, D-54292 Trier, Germany\\
}
\begin{abstract} 
A novel random field model or the reconstruction of turbulent velocity fluctuations 
from inhomogeneous characteristic flow quantities 
in terms of stochastic Fourier-type integrals 
has recently been introduced and analyzed by the authors.
This article concerns the numerical discretization and implementation of the model 
and discusses 
its key features
by means of numerical simulations. 
We present a suitable discretization scheme 
that combines a randomized quadrature method for stochastic integrals 
with a local linearization of the non-uniform 
advection of the turbulent structures by the mean flow.  
The convergence of the scheme towards the continuous model 
is verified analytically. 
Moreover, we describe 
an efficient algorithmic implementation 
that allows for flexible local evaluations of the simulated turbulence field. 
The main features of the 
model 
are illustrated by a variety of 
simulation results, each highlighting specific aspects such as 
the influence of the inhomogeneous model parameters 
on the generated fluctuations, 
spatio-temporal ergodicity properties under inhomogeneous flow conditions, 
and the validity of Kolmogorov's two-thirds law in dependence on the local turbulence Reynolds number.
\end{abstract}
\maketitle
{\sc AMS-Classification.}  60G60, 60H05, 65C05 76F55, 76M35







\section{Introduction}

Random field models for the reconstruction of turbulent velocity fluctuations 
from inhomogeneous characteristic flow quantities 
are utilized in many ways in various scientific applications, e.g., 
in the context of RANS and LES computations \cite{Wie+19, Pol+20, XYD22, Zwi+24}. 
While there has been significant progress in 
random field modeling of synthetic inhomogeneous turbulence in the past decades, 
see, e.g., \cite{BBLC94, BED03, SSST14, GJYZ23} and references therein,  
the analytical foundations of available 
inhomogeneous 
models tend to be of limited scope, which 
is often related to 
the use of semi-discrete modelling frameworks 
and a lack of separation 
of questions concerning modeling and analysis 
from questions concerning numerical approximation and simulation. 
Against this background, a 
flexible and 
fully continuous random field model for 
the reconstruction of inhomogeneous turbulence from RANS-type data 
that is accessible to 
a rigorous analytical validation of the desired model properties 
has been 
derived 
and analyzed 
by the authors in \cite{AKLMW24}. 
It is based on an explicit representation formula 
in terms of stochastic integrals 
combining 
moving average and spectral representations in time and space, respectively, 
along with a two-scale approach 
that allows to 
examine the inhomogeneous model properties by means of asymptotic analysis. 
Focusing on second-order statistics and adopting the simplifying assumption of Gaussianity, the model 
manages to consistently include 
spatio-temporal variations in turbulence length and time scales 
and in the 
turbulence Reynolds number, 
non-uniform advection of the turbulent structures by the mean flow, 
and  
the 
recovery of 
inhomogeneous model parameters 
such as turbulent kinetic energy, Reynolds stresses, and dissipation rate 
in terms of local averages of a single sample path in time or space.  
In the present article 
we address the numerical approximation and simulation 
of the model 
by developing a suitable discretization scheme,  
verifying its convergence towards the continuous model, 
and describing an efficient and flexible algorithmic implementation.  
In addition, we illustrate and discuss 
the key model features 
by means of a variety of 
instructive 
simulation results. 

As the literature background 
on random field modeling of inhomogeneous 
turbulence 
has already been portrayed in 
the accompanying paper \cite{AKLMW24}, 
we focus here on a short review of 
related works
concerning numerical discretization methods 
for given random field models before 
describing the contributions of this article in more detail. 
Of particular interest in the context of our model are 
discrete 
spectral representation methods 
involving finite sums of random Fourier modes 
\cite{Pri01, Sab13, DS25}, 
which naturally arise from 
disretizations of the 
classical continuous spectral 
decompositions  
of homogeneous random fields and their correlation functions 
\cite{GS74,MY75, Bre14}. 
Such methods have been widely used 
for the approximation of homogeneous random fields 
and can be interpreted as numerical quadrature techniques 
for stochastic Fourier-type integral 
representations. The 
quadrature points may be chosen 
deterministically 
\cite{Shi72, PS95, SD96, Mann98, VS21}  
or in a randomized way 
\cite{Kra70, Shi71, SJ72, Kur93, BK95, Kur95, KS06, KSK13}, 
where the latter approach offers the advantage of 
being able to exactly reproduce the correlation structure 
of the 
spectrally continuous 
target model 
while 
avoiding artificial 
periodicity 
related to the numerical disretization. 
Discrete spectral representation methods 
have been extended from early on 
to the approximation of inhomogeneous random field models,
see, e.g., \cite{Shi71, SJ72, PS95, LCS07}, 
and have turned out to be particularly useful 
for the generation of inhomogeneous synthetic turbulence fields 
from RANS-type data, see, e.g.,  
\cite{BBLC94, SSC01, BED03, SSST14, YB14, ADDK20, ADDK22,   GJYZ23}.
In the latter class of articles,  however, 
the random fields 
are typically formulated in 
a semi-discrete setting only, 
without 
a full 
specification of a continuous limit 
model and its analytic properties in view of inhomogeneous flow conditions. 
We remark that the randomized approach in \cite{ADDK20, ADDK22} 
is partly similar to the one employed in the present article,  
but does not address aspects such as 
temporal decay, non-uniform mean flow advection, 
and consistency with prescribed values of the dissipation rate. 
Further numerical methods besides the spectral representation approach 
include moving average techniques \cite{KSJ03, Ewe07, KCX13, Ewe16} and wavelet methods \cite{EM95, KKS07,KS08, DL18}, among others. 
Here we highlight the moving average approach in \cite{Ewe07, Ewe16}, 
which 
treats 
rather general 
inhomogeneous flow conditions 
and offers a 
considerable 
analytical foundation, separating 
questions related to a continuous limit model 
from questions related to the numerical approximation thereof. 
The approach 
is fundamentally different 
from the one employed here  
in that, e.g., it is based on a Lagrangian 
instead of an Eulerian perspective, 
it is not designed to 
enable 
flexible localized evaluations of the random field, 
and it leads to weaker regularity properties 
with regard to the temporal evolution.

The contributions and structure of this article are as follows. 
We begin by briefly recalling the random field model 
for the reconstruction of inhomogeneous turbulence 
from \cite{AKLMW24} in Section~\ref{sec:model}, 
before introducing  a suitable discretization scheme 
and analytically verifying its convergence in Section~\ref{sec:discretization_convergence}. 
The discretization combines 
a randomized quadrature method for stochastic integrals 
with a local linearization of the non-uniform mean flow advection. 
The randomized quadrature 
corresponds to  
replacing the underlying white noise random measure -- 
defined on the product space of time domain and spatial Fourier domain -- 
by a finite sum of randomly distributed 
Dirac measures 
with random weights, 
so that the resulting discretized model 
is essentially given 
in terms of  
a finite sum of 
inhomogeneously scaled random traveling plane waves 
that are superposed according to a 
moving average construction in time. 
To facilitate efficient local evaluations at 
different 
time points, 
a stratification approach \wrt the time variable is employed 
for the specification of the random quadrature points. 
The convergence of the discretization scheme 
towards the continuous model 
is established in Theorem~\ref{thm:conv} 
using 
the two-scale approach 
of the model  
in combination with 
the isometry property of stochastic integrals and 
a general convergence result for stratified randomized quadrature methods 
established 
in the appendix of the article. 
In addition, Corollary~\ref{prop:flow_properties} ensures the validity of characteristic flow properties for the discretized field in analogy to the respective findings for the continuous model. 
Section~\ref{sec:implementation} is concerned with 
the algorithmic implementation of the 
discretized model, 
where we are specifically interested in 
enabling flexible 
and efficient 
local evaluations of the 
random field  
that do not necessarily require 
global simulations 
on larger subsets of the spatio-temporal computational domain. 
This aspect is particularly essential 
in the context of highly resolved turbulence-driven dynamics, 
e.g., of particles or fibers 
\cite{MW11,HMRW13,Wie+19}, 
where the relevant 
evaluation points for the velocity field 
are not known in advance but are determined dynamically in the course of the simulation. 
We 
propose an efficient sampling algorithm that ensures 
flexibility \wrt localized evaluations of the velocity field  
while maintaining 
consistency \wrt the random numbers employed in each evaluation. 
The relevant random numbers depend on 
the support of the involved time integration kernel, 
which 
in turn 
depends on the local turbulence scales at the considered evaluation point. 
The algorithm manages to keep track of those random numbers used in previous evaluations 
that are 
potentially relevant for the ongoing simulation process. 
In Section~\ref{sec:sim_res} we present a variety of simulation results 
in order to illustrate and discuss the key features of 
our 
model, 
which have previously been addressed 
in \cite{AKLMW24} from an analytical point of view. 
Subsection~\ref{sec:parameters} concerns the significance of the model parameters 
and shows how inhomogeneities in the local turbulence scales 
and the local turbulence Reynolds number 
are reflected in 
single realizations of the turbulent velocity field. 
The aspect 
of non-uniform mean flow advection is visualized and discussed as well. 
The spatio-temporal ergodicity properties of the 
model 
are illustrated in Subsection~\ref{sec:ergodicity} 
by demonstrating that the 
inhomogeneous 
characteristic flow quantities 
of turbulent kinetic energy and dissipation rate 
can be recovered 
in terms of 
local averages of a singe sample path.
The validity of Komogorov's two-thirds law for the spatial velocity increments in the inertial range 
and its essential dependence  on the local turbulence Reynolds number are demonstrated in 
the final 
Subsection~\ref{sec:Kolmogorov}.  
Appendix~\ref{app:wn} and Appendix~\ref{app:approx} provide 
some background on white noise random measures 
and stochastic integration 
as well as  
suitable abstract auxiliary results on randomized quadrature methods 
for stochastic integrals.

\subsection*{General notation and conventions}

We set  $\R^+ = (0,\infty)$, $\R^+_0 = [0,\infty)$ 
and write $\Re z$ and $\Im z$ for the real and the imaginary part of a  complex number $z\in\C$.
Unless specified otherwise, 
$d$, $\ell$, $m$, $n\in\N$ 
are assumed to be arbitrary natural numbers. 
Vectors and matrices are denoted by small bold letters and capital bold letters, respectively. 
Basic tensor operations are given by  
$\bm{a} \cdot \bm{b} = \sum_{j} a_j b_j$, 
$\bm{a} \otimes \bm{b} = ( a_j b_k)_{j,k}$, 
$\bm{A} \cdot \bm{b} = (\sum_{k} A_{j,k} b_k)_j$, 
$\bm{A} \cdot \bm{B} = (\sum_{k} A_{j,k} B_{k,l})_{j,l}$, 
where $\bm{a} = (a_j)_j$, $\bm{b} = (b_j)_j$, $\bm{A} = (A_{j,k})_{j,k}$, $\bm{B} = (B_{j,k})_{j,k}$ are vectors and matrices of suitable dimensions.  
We write $\| \bm{a} \| = (\sum_j |a_j|^2)^{1/2}$ and $\| \bm{A} \| = (\sum_{j,k} |A_{j,k}|^2)^{1/2}$ for the Euclidean norm and the Frobenius norm of real or complex vectors and matrices, respectively, 
and $S^2 = \{ \x \in \R^3 : \|\x\| = 1 \}$ for the unit sphere in $\R^3$.
The Borel $\sigma$-algebra on a given Borel set $U \subset \R^n$ 
is denoted by $\B(U)$, 
$\lambda^n$ stands for the Lebesgue measure on $(\R^n,\B(\R^n))$, 
and $U_{S^2}$ denotes the uniform distribution on $(S^2, \B(S^2))$ (i.e., the normalized surface measure on the $2$-sphere). 
For any measure space $(U,\mathcal{A},\mu)$ and any finite-dimensional normed vector space $V$, we use the notation $L^2(\mu;V) = L^2(U,\mu;V)$ for the space of (equivalence classes of) measurable and square-integrable functions from $U$ to $V$. 
All random variables and random fields are assumed to be defined on the same underlying probability space $(\Omega,\mathscr{F},\Pr)$. 
The expected value of an integrable random variable 
$X \colon \Omega \to V$  is denoted by $\E[X] = \int_\Omega X \,\d\Pr$.


\section{Inhomogeneous turbulence model}\label{sec:model}

We briefly recall the random field model for the reconstruction of inhomogeneous turbulence from  
characteristic flow quantities 
that has been developed and analyzed in detail in \cite{AKLMW24}. 
Building on the theory of homogeneous turbulence, the inhomogeneous model is based on 
an asymptotic two-scale approach separating the turbulent fluctuations (micro scale) from macro scale variations of the flow quantities. 

Proceeding from a macro length $x_0$ associated with the geometry of the flow problem as well as from typical values for kinematic viscosity $\nu_0$, turbulent kinetic energy $k_0$, and dissipation rate $\eps_0$ as reference values for the non-dimensionalization, we consider 
\begin{align*}
u_0 = \sqrt{k_0}, \qquad t_0 = x_0/\sqrt{k_0}, \qquad x_\mu = \sqrt{k_0}^3/\eps_0, \qquad t_\mu = k_0/\eps_0. 
\end{align*} 
In contrast to the macro scale specified by the macro length $x_0$ and respective time $t_0$,
the quantities $x_\mu$ and $t_\mu$ represent the typical turbulent length and time and indicate a micro scale associated with the turbulent fluctuations.
The velocity $u_0$ is chosen such that $u_0 = x_0/t_0 = x_\mu/t_\mu$.
In inhomogeneous turbulence, the characteristic flow quantities of mean velocity $\overline{\u}(\x,t)$, kinematic viscosity $\nu(\x,t)$, turbulent kinetic energy $k(\x,t)$, and dissipation rate $\eps(\x,t)$ 
are functions of space $\x$ and time $t$. We assume that these functions are given in dimensionless form viewed from a macroscopic perspective, i.e., the function values are scaled with the reference values $u_0$, $\nu_0$, $k_0$, $\eps_0$, while the arguments $\x$, $t$ refer to the macro scale $x_0$, $t_0$. The dimensionless turbulent velocity is expressed by a Reynolds-type decomposition in terms of mean velocity and turbulent fluctuations,
\begin{align}\label{eq:decomp_u}
\u(\x, t) = \overline{\u}(\x, t) + \u'(\x, t),
\qquad \E\bigl[ \u(\x, t) \bigr]  = \overline{\u}(\x, t) , \quad \E\bigl[ \u'(\x, t) \bigr] = \bm{0},
\end{align} 
where we interpret the turbulent fluctuations $\u'$ as a random field depending on the characteristic flow quantities.  All further dependencies of $\u'$ are covered by two dimensionless numbers
\begin{align}\label{eq:z_delta}
z = \frac{\eps_0\nu_0}{k_0^2},\qquad \delta = \frac{x_\mu}{x_0} = \frac{t_\mu}{t_0} =\frac{\sqrt{k_0}^3}{\eps_0 x_0},
\end{align}
$z, \delta\ll 1$. 
The  parameter $z$ is known from the homogeneous turbulence theory 
and represents the inverse of the turbulence Reynolds number. 
It is proportional to the inverse of the turbulent viscosity ratio and indicates the scale ratio between the turbulent fine-scale structure (Kolmogorov scales of small vortices dissipating into heat) and the turbulent large scale structure (large energy-bearing vortices). 
The parameter $\delta$ entering via the two-scale approach represents the ratio between the turbulent scale (micro scale) and the macro scale associated with the geometry of the flow problem and is referred to as turbulence scale ratio. 
While the decomposition \eqref{eq:decomp_u} is formulated from the perspective of the macro scale $x_0$, $t_0$, the modeling of the turbulent fluctuations $\u'$ is based on the micro scale $x_\mu$, $t_\mu$. Scaling factors allow locally the adjustment of the values $x_\mu$, $t_\mu$, $u_0$ and $z$ to the spatio-temporal variations of the characteristic flow quantities, i.e.,
\begin{equation}\label{eq:sigma_x_t_u_z}
\begin{aligned}
\sx(\x,t) = \!\frac{\sqrt{k(\x,t)}^{3}}{\eps(\x,t)}, \quad \st(\x,t) = \frac{k(\x,t)}{\eps(\x,t)}, \quad \su(\x,t) = \sqrt{k(\x,t)}, \quad \sz(\x,t) = \frac{\eps(\x,t)\nu(\x,t)}{k(\x,t)^2}.
\end{aligned}
\end{equation}
The inhomogeneous random field model for the turbulent fluctuations is then specified as follows. 

\begin{model}\label{model:turb_inhom}
An \emph{inhomogeneous turbulence field} $\u' = (\u'(\x,t))_{(\x,t) \in \R^3 \times \R}$ is described as a centered, $\R^3$-valued Gaussian random field of the form 
\begin{equation}\label{eq:u'}
\begin{aligned}
\u'(\x,t)   =  \su(\x,t)\;\Re \! &\int_{\R^+ \times S^2 \times \R}  \Bigl(\frac1{\delta\st(\x,t)}\Bigr)^{1/2}\eta\Bigl(\frac{1}{\delta\st(\x,t)}(t-s)\Bigr) 
\exp\Bigl\{ i \frac{1}{\delta}\;\! \kappa\;\! \btheta \cdot \flow(s;\x,t) \Bigr\} 
\\ & 
\qquad 
\vphantom{\int} 
\sx^{1/2}(\x,t) \;\! E^{1/2} \bigl(\sx(\x,t)\kappa; \sz(\x,t)z\bigr) \;\!  \Pt \cdot \bm{L}(\x,t) \cdot \wn(\d\kappa,\d\btheta,\d s),
\end{aligned}
\end{equation}
where the characteristic numbers $\delta, z \in \R^+$ and the scaling functions $\sx$, $\st$, $\su$, $\sz \colon \R^3 \times \R \to \R^+$ are given by \eqref{eq:z_delta} and \eqref{eq:sigma_x_t_u_z}, and the following holds:
\begin{enumerate}
\item The energy spectrum function 
$\R^+\times\R^+ \ni (\kappa,\zeta) \mapsto E(\kappa;\zeta) \in \R_0^+$ 
is  continuously differentiable and fulfills the integral conditions 
\begin{equation}\label{eq:energy_spectrum}
\int_0^\infty E(\kappa;\zeta) \,\d\kappa = 1, \qquad \int_0^\infty \kappa^2 E(\kappa;\zeta) \,\d\kappa = \frac{1}{2\zeta}.
\end{equation} 
\item The time integration kernel 
$\eta \colon \R \to \R$ is continuously differentiable, has compact support, and satisfies
\begin{equation}\label{eq:basic_integral_prop_eta}
\int_\R \eta^2(s) \,\d s = 1.
\end{equation}
\item For every $(\x,t) \in \R^3 \times \R$ the mean flow function $\flow(\bdot;\x,t) \colon \R \to \R^3$ is continuous and solves the integral equation 
\begin{equation}\label{eq:mean_flow}
\flow(s;\x,t) = \x + \int_t^s\overline{\u} \bigl( \flow(r;\x,t), r \bigr) \,\d r, \quad s \in \R.
\end{equation}
\item The flow quantities $k,\eps,\nu \colon \R^3 \times \R \to \R^+$ and $\overline{\u}\colon \R^3 \times \R \to \R^3$ are continuous in $(\x,t)$, differentiable in $\x$, and such that $\nabla_\x k$, $\nabla_\x \eps$, $\nabla_\x \nu$ and $\nabla_\x \overline{\u}$ are continuous in $(\x,t)$. 
The mean velocity gradient $\nabla_\x \overline{\u} \colon \R^3 \times \R \to \R^{3 \times 3}$ is 
bounded. 
\vspace{1mm}
\item 
The anisotropy function $\bm{L}:\R^3\times \R \to \R^{3\times 3}$ is continuous in $(\x,t)$, differentiable in $\x$, and such that $\nabla_{\x}\bm{L}$ is continuous in $(\x,t)$. Moreover, it satisfies 
$\|\bm{L}(\x,t)\|^2=3$. 
\vspace{1mm}
\item 
$\wn$ is a $\C^3$-valued 
Gaussian white noise on $\R^+ \times S^2 \times \R$ with structural measure  $2\lambda^1|_{\R^+}\otimes U_{S^2}\otimes\lambda^1$ in the sense of Definition \ref{def:white_noise}, 
where $U_{S^2}$ 
and $\lambda^1$ denote the uniform distribution on the $2$-sphere 
(i.e., the normalized surface measure) and the one-dimensional Lebesgue measure.
\end{enumerate}
\vspace{1mm}
The matrix $\Pt \in \R^{3 \times 3}$ in \eqref{eq:u'} denotes the projector onto the orthogonal complement of $\btheta$, $\btheta \in S^2 = \{\x\in\R^3\colon \|\x\| = 1\}$ ($2$-sphere). 
In addition, the following technical integrability condition related to the spatial mean-square differentiability of $\u'$ is fulfilled:
\vspace{1mm}
\begin{enumerate}
\item[\textbf{g)}] For every $(\x,t) \in \R^3 \times \R$ there exists an $r \in \R^+$ such that
\begin{equation*}
\int_0^\infty \sup_{(\y,s) \in B_r(\x,t)} \Bigl\| \nabla_{\y} \Bigl( \sx^{1/2}(\y,s) E^{1/2}\bigl(\sx(\y,s) \kappa; \sz(\y,s)z \bigr) \Bigr) \Bigr\|^2 \d \kappa < \infty, 
\end{equation*}
where $B_r(\x,t)$ denotes the closed ball in $\R^3\times\R$ with center $(\x,t)$ and radius $r$.
\end{enumerate}
\end{model}
\vspace{.5mm}

Characteristic flow properties 
\wrt turbulent kinetic energy, Reynolds stresses, dissipation rate, incompressibility, 
and spatio-temporal ergodicity of the inhomogeneous turbulence field $\u'$ have been established in \cite[Theorems~5.5 and 5.6]{AKLMW24} 
in the form of asymptotic results for the turbulence scale ratio $\delta\to 0$. The assumptions in Model~\ref{model:turb_inhom} in particular ensure that the random field $\u'$ is mean-square differentiable \wrt $\x$, so that the gradient $\nabla_\x \u'$ is well-defined in the mean-square sense. 
In addition, both $\u'$ and $\nabla_\x \u'$ are mean-square continuous in $(\x,t)$.
Concrete examples for the energy spectrum $E(\kappa;\zeta)$ and the time integration kernel $\eta(s)$ 
can be found in \cite[Examples~2.3, 2.4, and 5.3]{AKLMW24}. 
The anisotropy factor $\bm{L}(\x,t)$ can be used to control directional weightings of the one-point velocity correlations (Reynolds stresses), see  \cite[Theorem~5.5]{AKLMW24}, and simplifies to the identity matrix 
$\bm{L}(\x,t)=\bm{I}$ in the case of isotropic one-point velocity correlations. 
Unlike in \cite[Model~5.1]{AKLMW24}, here the stochastic integral representation of $\u'$ is equivalently formulated in terms of spherical coordinates instead of Cartesian coordinates, as this facilitates the establishment of a discretization scheme in Section~\ref{sec:discretization_convergence}. Furthermore, to simplify the exposition we consider a slightly less general setting by assuming that the time integration kernel $\eta$ has a compact support. Note, however, that the discretization scheme and the convergence results presented in this article can be easily extended to kernels with non-compact support, using additional truncation and approximation arguments.

\section{Discretization scheme and convergence analysis}
\label{sec:discretization_convergence}

In this section we 
introduce  
and analyze 
a suitable discretization scheme for 
the numerical approximation of 
our random field model of inhomogeneous turbulence. 
Analytical results concerning the convergence of the scheme 
and 
characteristic flow properties of the discretized model are presented in 
Theorem~\ref{thm:conv} and Corollary~\ref{prop:flow_properties} below. 

For the discretization of Model~\ref{model:turb_inhom} we employ a stratified Monte Carlo quadrature method 
in order 
to approximate the 
stochastic 
integral 
appearing 
in the representation formula \eqref{eq:u'}. 
Formally, this corresponds to replacing the white noise term 
$\wn$ 
in \eqref{eq:u'} 
by a discrete random measure 
$\wn_N$ 
consisting of randomly distributed point masses on $\R^+ \times S^2 \times \R$ with 
random weights, where $N\in\N$ is a discretization parameter.
In combination with 
a linearization 
of the mean flow function $\flow$ 
in \eqref{eq:u'}, 
this leads to 
approximating random fields 
of the form 
\begin{align}\label{eq:u'N_int}
\u'_N(\x,t)  = \su(\x,t)\; \Re & \!\int_{\R^+ \times S^2 \times \R}   \Bigl(\frac1{\delta\st(\x,t)}\Bigr)^{1/2} \notag
\!\eta
\Bigl(\frac{1}{\delta\st(\x,t)}(t-s)\Bigr)
\exp\Bigl\{ i \frac{1}{\delta}\;\! \kappa\;\! \btheta \cdot \bigl(\x-(t-s)\overline\u(\x,t)\bigr) \Bigr\}
\\ & 
\quad 
\vphantom{\int}
\sx^{1/2}(\x,t) \;\! E^{1/2} \bigl(\sx(\x,t)\kappa; \sz(\x,t)z\bigr)  \;\!  \Pt \cdot \bm{L}(\x,t) \cdot \wn_N(\d\kappa,\d\btheta,\d s).
\end{align} 
The 
structure of the 
$\C^3$-valued random measure $\wn_N$ 
involves a 
partitioning of the domain of integration  
$\R^+ \times S^2 \times \R$ 
into strata of the form $\R^+ \times S^2 \times I_j$,
 where $I_j=[j\Delta s,(j+1)\Delta s)$ are stratification intervals of fixed length  
$\Delta s$.  
It is given by 
\begin{align}\label{eq:WN}
\wn_N(\;\!\cdot\;\!) = 
\sum_{j\in\mathbb Z} \frac{1}{\sqrt{N}} 
\sum_{n=1}^N 
\Bigl(\frac{\Delta s}{ p(\kappa_{jn}) }\Bigr)^{1/2} 
\wn_{jn} \, \delta_{(\kappa_{jn}, \btheta_{jn}, s_{jn})}(\;\!\cdot\;\!) 
\end{align}
with suitable $\C^3$-valued random variables $\wn_{jn}$ and random quadrature points  $(\kappa_{jn}, \btheta_{jn}, s_{jn})$ in $\R^+ \times S^2 \times I_j$, where 
$\kappa_{jn}$ is drawn from the reference distribution  
$\ind_{\R^+}(\kappa) \,p(\kappa) \,\d\kappa$. 
Here 
$p$ is a reference probability density function on $\R^+$ 
and $\delta_{(\kappa, \btheta, s)}$ denotes Dirac measure at  
$(\kappa, \btheta, s)$. 
We refer to Appendix~\ref{app:approx} for 
a presentation 
of general auxiliary results on Monte Carlo quadrature 
methods 
for white noise integrals motivating the specific choice \eqref{eq:WN}. 

Rewriting the discretized stochastic integral in \eqref{eq:u'N_int} as a random sum, 
the numerical approximation scheme is summarized and specified as follows. 

\begin{discretization}\label{discretization}
Let $\u'=(\u'(\x,t))_{(\x,t)\in\R^3\times \R}$ be an inhomogeneous turbulence field in the sense of Model~\ref{model:turb_inhom}. 
For $N\in\N$ the approximating random field $\u'_N=(\u'_N(\x,t))_{(\x,t)\in\R^3\times \R}$ is defined by \eqref{eq:u'N_int} and \eqref{eq:WN}, i.e., 
\begin{equation}\label{eq:uN'2}
\begin{aligned} 
\u'_N(\x,t) 
= 
\su(\x,t)\;\Re\sum_{j\in\mathbb Z}  \frac{1}{\sqrt{N}} & \sum_{n=1}^N 
\Bigl(\frac1{\delta\st(\x,t)}\Bigr)^{1/2} \eta\Bigl(\frac{1}{\delta\st(\x,t)}(t-s_{jn})\Bigr)
\\ & 
\quad
\exp\Bigl\{i\frac1\delta\;\!\kappa_{jn}\btheta_{jn}\cdot\bigl(\x-(t-s_{jn})\overline\u(\x,t)\bigr) \Bigr\} 
\Bigl(\frac{\Delta s}{ p(\kappa_{jn}) }\Bigr)^{1/2} 
\\ & 
 \quad
 \vphantom{\int}
\sx^{1/2}(\x,t) \;\! E^{1/2} \bigl(\sx(\x,t)\kappa_{jn}; \sz(\x,t)z\bigr) \;\! 
\P(\btheta_{jn})\cdot \bm{L}(\x,t) \cdot \wn_{jn},
\end{aligned}
\end{equation}
where the following is assumed:    
\begin{itemize}
\item
The 
random wave numbers 
$\kappa_{jn}$, $j\in\mathbb Z$, 
$n\in\N$, 
are independent and identically distributed according to 
the 
probability density function $p\colon \R^+\to\R^+_0$,
which satisfies 
for every $(\x,t)\in\R^3\times\R$ 
the support condition 
\begin{equation}\label{eq:cond_p}
\lambda^1\bigl( \bigl\{\kappa\in\R^+ \colon p(\kappa)=0 \text{ and } E\bigl(\sx(\x,t)\kappa; \sz(\x,t)z\bigr)\neq 0 \bigr\}\bigr) = 0. 
\vspace{1mm}
\end{equation}
\item
The random orientation vectors $\btheta_{jn}$ $j\in\mathbb Z$, 
$n\in\N$, 
are 
independent and identically distributed according to the uniform distribution on the unit sphere $S^2$.
\vspace{1mm}
\item
The temporal quadrature points $s_{jn}$, $j\in\mathbb Z$, 
$n\in\N$, 
are independent random variables such that each $s_{jn}$ is uniformly distributed on the interval $I_j=[j\Delta s,(j+1)\Delta s)$, where $\Delta s$ is a fixed stratification length.
\vspace{1mm}
\item
The 
complex 
noise vectors  
$\wn_{jn}$ 
are 
square-integrable and 
such that 
the 
$\R^3$-valued random variables $\Re\;\!\wn_{jn}$, $\Im\;\!\wn_{jn}$, $j\in\mathbb Z$, 
$n\in\N$, 
are 
independent and identically distributed with mean zero 
and identity covariance matrix 
\begin{align*}
\qquad \E\bigl[\Re\;\!\wn_{jn}\bigr]=\E\bigl[\Im\;\!\wn_{jn}\bigr]=\bm0, \quad \E\bigl[\Re\;\!\wn_{jn}\otimes\Re\;\!\wn_{jn}\bigr]=\E\bigl[\Im\;\!\wn_{jn}\otimes\Im\;\!\wn_{jn}\bigr]=\bm I.
\end{align*}
\end{itemize}
In addition to the independence assumptions above, the combined family of random variables $\kappa_{jn}$, $\btheta_{jn}$, $s_{j,n}$, $\wn_{j,n}$, $j \in\mathbb Z$, $n\in\N$, is assumed to be independent as well.
\end{discretization}

Observe 
that the sum 
$\sum_{j\in \mathbb Z}$ 
appearing 
in
the discretization formula \eqref{eq:uN'2} is actually finite.
Indeed, 
for every fixed choice of $(\x,t)$, all but finitely many of the summands vanish 
due to the boundedness of the support of the time integration kernel $\eta$. 
This is exploited in the
algorithmic implementation of the 
discretized model 
described 
in Section~\ref{sec:implementation} below.
A 
natural choice for the reference density $p$ is 
given 
by 
the energy spectrum function 
without scaling factors, i.e., 
$p(\kappa)=E(\kappa;z)$,
provided 
that the support condition \eqref{eq:cond_p} is satisfied.
The 
condition is trivially fulfilled 
for this choice  
if the employed energy spectrum is strictly positive. 
It is also worth noting that various alternative methods 
as well as 
variants 
of the proposed method 
for the discretization of Model~\ref{model:turb_inhom} 
are conceivable. 
For instance, the stratification approach \wrt the integration variable $s$ 
may be extended to the integration variables $\kappa$ and $\btheta$ 
by employing suitable partitions of the respective domains $\R^+$ and $S^2$; 
see, e.g., \cite[Section~7.1]{KS06} for 
related considerations 
in the context of homogeneous
fields. 
We 
focus on
the discretization scheme introduced above in order to keep the length of the manuscript within reasonable bounds.

For the sake of presentation it is convenient to introduce the shorthand notation
\begin{equation}\label{eq:def_etabb_fxbb}
\begin{aligned}
\etabb(s;\x,t) & = \Bigl(\frac1{\delta\st(\x,t)}\Bigr)^{1/2}\eta\Bigl(\frac{1}{\delta\st(\x,t)}s\Bigr),
\\ \mathbbm{e}(\kappa;\x,t) & = \sx^{1/2}(\x,t) \;\! E^{1/2} \bigl(\sx(\x,t) \kappa; \sz(\x,t)z \bigr),
\\ \Pbb(\btheta;\x,t) & = \su(\x,t) \;\! \Pt \cdot \bm{L}(\x,t), 
\end{aligned}
\end{equation}
so that the 
representation formulas \eqref{eq:u'} and \eqref{eq:uN'2} in Model~\ref{model:turb_inhom} and Discretization~\ref{discretization} can 
be rewritten as
\begin{align}
\u'(\x,t) & = 
\Re \! \int_{\R^+\times S^2 \times \R} 
\etabb(t-s;\x,t) \;\! \exp\Bigl\{ i \frac{1}{\delta}\;\!\kappa\;\! \btheta\cdot \flow(s;\x,t) \Bigr\} \;\!  
\mathbbm{e}(\kappa;\x,t) \,  \Pbb(\btheta;\x,t) \cdot \wn(\d\kappa, \d\btheta, \d s), \notag
\\
\u'_N(\x,t)  &= \Re \!\int_{\R^+ \times S^2 \times \R}  
\etabb(t-s;\x,t) \;\! \exp\Bigl\{ i \frac{1}{\delta}\;\! \kappa\;\! \btheta \cdot \bigl(\x-(t-s)\overline\u(\x,t)\bigr) \Bigr\}  \notag
\\ &
\qquad\qquad\qquad\qquad\qquad\qquad\qquad\qquad\qquad\qquad \qquad\,
\ebb(\kappa;\x,t) \, \Pbb(\btheta;\x,t) \cdot \wn_N(\d\kappa,\d\btheta,\d s),  \notag
\end{align}
where $\wn_N$ is the discrete random measure defined in \eqref{eq:WN}. 

In Theorem~\ref{thm:conv} below we 
investigate the convergence behaviour of the proposed numerical scheme and 
analytically justify that the
discretized fields $\u'_N$ may be used as approximations of the inhomogeneous turbulence field $\u'$. 
Similar to the analysis of the continuous model \cite{AKLMW24}, 
the turbulence scale ratio $\delta\ll 1$ 
introduced in \eqref{eq:z_delta}
plays a crucial role in this context 
as it allows for an asymptotic control of the macroscopic variations of the underlying flow quantities. 
We proceed in two steps and first show that the auxiliary random field
\begin{equation}\label{eq:u'aux}
\begin{aligned}
\u'_{\mathrm{aux}}(\x,t)  = \Re \!\int_{\R^+ \times S^2 \times \R}  \etabb(t-s;\x,t)  \exp\Bigl\{ i \frac{1}{\delta}\;\! \kappa\;\! \btheta \cdot 
 \bigl(\x-(&t-s)\overline\u(\x,t)\bigr) 
\Bigr\}
\\ & \quad
\ebb(\kappa;\x,t) \, \Pbb(\btheta;\x,t) \cdot \wn(\d\kappa,\d\btheta,\d s),
\end{aligned}
\end{equation}
approximates $\u'$ in the mean-square sense as $\delta\to0$. 
This accounts for the linearization of the mean flow function $\flow$ modelling the advection of the turbulent fluctuations.
In a second step 
we show for every fixed value of $\delta$ that the 
finite-dimensional distributions of the  
discretized field $\u'_N$ 
converge to 
the respective finite-dimensional distributions 
of the auxiliary field 
$\u'_{\mathrm{aux}}$ as $N\to\infty$.
This addresses 
the discretization of the underlying white noise $\wn$.

\begin{theorem}[Convergence]\label{thm:conv}
Let $\u'=(\u'(\x,t))_{(\x,t)\in\R^3\times \R}$ be an inhomogeneous  
turbulence field in the sense of Model~\ref{model:turb_inhom}, 
let  $\u'_N=(\u'_N(\x,t))_{(\x,t)\in\R^3\times \R}$, $N\in\N$, 
be 
the
family of 
approximating random fields  
specified in  
Discretization~\ref{discretization},
and let  $\u_{\mathrm{aux}}'=(\u'_{\mathrm{aux}}(\x,t))_{(\x,t)\in\R^3\times \R}$ 
be 
given by 
\eqref{eq:u'aux}. 
Then the following assertions hold: 
\begin{enumerate}
\item 
The auxiliary field $\u_{\mathrm{aux}}'$ and its scaled gradient $\delta\;\!\nabla_\x \u_{\mathrm{aux}}'$
approximate $\u'$ and $\delta\;\!\nabla_\x \u'$ in the mean-square sense 
as the turbulence scale ratio $\delta$ tends to zero. 
More precisely, for every 
compact subset $K$ of $\R^3\times\R$ 
we have that 
\begin{equation*}
\lim_{\delta \to 0 } 
\sup_{(\x,t)\in K}
\E\Bigl[ 
\bigl\| \u'(\x,t) - \u'_{\mathrm{aux}}(\x,t) \bigr\|^2 
+ 
\bigl\| \delta\;\!\nabla_\x \u'(\x,t) - \delta\;\!\nabla_\x \u'_{\mathrm{aux}}(\x,t) \bigr\|^2 
\Bigr]
= 0.
\end{equation*} 
\item
For every value of  
$\delta$,  
the finite-dimensional distributions of $\u'_N$ and $\nabla_\x\u'_N$ converge 
weakly 
to the respective 
finite-dimensional distributions of $\u'_{\mathrm{aux}}$ and $\nabla_\x\u'_{\mathrm{aux}}$ as 
$N\to\infty$, i.e.,   
for any choice of points $(\x_1,t_1)$, $(\x_2,t_2)$, \ldots, 
$(\x_\ell,t_\ell)\in\R^3\times\R$, $\ell\in\N$, 
it holds that 
\begin{equation*}
\bigl(\u'_N(\x_1,t_1), 
\ldots,\u'_N(\x_\ell,t_\ell)\bigr) 
\, \xrightarrow{\;\mathrm{d}\;} \,
\bigl(\u'_{\mathrm{aux}}(\x_1,t_1), 
\ldots,\u'_{\mathrm{aux}}(\x_\ell,t_\ell)\bigr). 
\end{equation*}
and 
\begin{equation*}
\bigl(\nabla_\x\u'_N(\x_1,t_1), 
\ldots,\nabla_\x\u'_N(\x_\ell,t_\ell)\bigr) 
\, \xrightarrow{\;\mathrm{d}\;} \,
\bigl(\nabla_\x\u'_{\mathrm{aux}}(\x_1,t_1), 
\ldots,\nabla_\x\u'_{\mathrm{aux}}(\x_\ell,t_\ell)\bigr). 
\end{equation*}
Moreover, 
for every $N\in\N$ the covariance structure of the random field  $\u'_N$ is 
identical to the covariance structure of $\u'_{\mathrm{aux}}$. 
\end{enumerate}
\end{theorem}

\begin{proof}
In order to verify the assertion of part a), 
first note 
that the isometric property of the stochastic integral \eqref{eq:isometric_property} 
implies the identity
\begin{equation}\label{eq:isometry_on_u'-u'_aux}
\begin{aligned}
& \E\Bigl[ \bigl\| \u'(\x,t) - \u'_{\mathrm{aux}}(\x,t) \bigr\|^2 \Bigr]
\\ & \qquad  = 
\int_{\R^+} \int_{S^2} \int_{ \R} \, \biggl| \, \exp\Bigl\{ i \frac{1}{\delta} \;\! \kappa \;\! \btheta \cdot \flow(s;\x,t)\Bigr\} 
- \exp\Bigl\{ i \frac{1}{\delta} \;\! \kappa \;\! \btheta \cdot \bigl( \x - (t-s) \overline{\u}(\x,t)\bigr) \Bigr\} \biggr|^2
\\ & \qquad \phantom{ = \int_{\R^+} \int_{S^2} \int_{ \R}} \;\; \bigl| \etabb(t-s;\x,t) \, \ebb(\kappa;\x,t) \bigr|^2 \bigl\| \Pbb(\btheta;\x,t) \bigr\|^2 
\,\d s\; U_{S^2}(\d\btheta)\, \d \kappa,
\end{aligned}
\end{equation}
where $U_{S^2}$ denotes the normalized surface measure on $S^2$, compare Model \ref{model:turb_inhom}. 
Considering the difference of the arguments of the 
two exponential functions  
in \eqref{eq:isometry_on_u'-u'_aux}, 
a straightforward calculation shows that 
\begin{align}\label{eq:estimate_xi_linear_xi}
\bigl\| \flow(s;\x,t) - \bigl( \x - (t-s) \overline{\u}(\x,t) \bigr) \bigr\| 
\leq |t-s| 
 \, \sup_{\tau}\, \bigl\| \overline{\u}\bigl( \flow(\tau;\x,t),\tau \bigr) - \overline{\u}(\x,t) \bigr\|,
\end{align}
where the supremum is taken over all $\tau \in [\min(s,t), \max(s,t)]$. 
For the sake of controlling the influence of the factor $1/\delta$ in \eqref{eq:isometry_on_u'-u'_aux}, 
take $C>0$ 
such that the support of $\eta$ lies in the interval $[-C,C]$ 
and 
observe 
that $\etabb$ has the beneficial property that $\etabb(t-s;\x,t) \neq 0$ implies $|t-s| \leq C \delta \st(\x,t)$. 
Combining this, \eqref{eq:isometry_on_u'-u'_aux}, \eqref{eq:estimate_xi_linear_xi}, 
and the fact that $\int_{S^2} \| \Pt \cdot \bm{L}(\x,t) \|^2\, U_{S^2}(\d\btheta) = 2$ 
yields 
\begin{align*}
\E\Bigl[ \bigl\| \u'(\x,t) - \u'_{\mathrm{aux}}(\x,t) \bigr\|^2 \Bigr] &
\leq  C^2 \st(\x,t)^2  \sup_{\tau}\, \bigl\| \overline{\u}\bigl( \flow(\tau;\x,t),\tau \bigr) - \overline{\u}(\x,t) \bigr\|^2\\ & \qquad 2\,\su(\x,t)^2 \int_{\R} \bigl| \etabb(t-s;\x,t) \bigr|^2 \d s \, \int_{\R^+} \kappa^2 \bigl| \ebb(\kappa;\x,t) \bigr|^2 \d\kappa ,
\end{align*}
where the supremum is taken over all $\tau \in [t-C\delta \st(\x,t), t + C\delta\st(\x,t)]$. 
Further employing the identities 
$\int_{\R^+} \kappa^2 | \ebb(\kappa;\x,t) |^2 \d\kappa = 1/(2 \sz(\x,t) z \sx(\x,t)^2)$ and $\int_\R |\etabb(t-s;\x,t)|^2 \d s = 1$ 
following from \eqref{eq:energy_spectrum} and \eqref{eq:basic_integral_prop_eta}, we conclude that
\begin{align} \label{eq:proof_convergence}
\E\Bigl[ \bigl\| \u'(\x,t) - 
\u'_{\mathrm{aux}}(\x,t) 
\bigr\|^2 \Bigr]\leq 
\frac{ C^2 }{\sz(\x,t) z} \sup_{\tau}\, \bigl\| \overline{\u}\bigl( \flow(\tau;\x,t),\tau \bigr) - \overline{\u}(\x,t) \bigr\|^2.
\end{align}
The uniform continuity of  $(\tau,\x,t) \mapsto \overline{\u}(\flow(\tau;\x,t),\tau)$ on compact subsets of $\R\times \R^3 \times \R$ and  the positivity and continuity of $\sz$ 
therefore imply that 
the mean-square error on the left-hand side of \eqref{eq:proof_convergence} 
converges to zero 
as $\delta\to0$, uniformly with respect to  $(\x,t)$ in any fixed compact subset 
of $\R^3\times\R$.
To complete the proof of part a), 
a corresponding convergence for 
$\delta \nabla_\x \u'(\x,t) - \delta \nabla_\x \u'_{\mathrm{aux}}(\x,t)$ in place of 
$\u'(\x,t) - \u'_{\mathrm{aux}}(\x,t)$ can be shown in a similar way. 
Indeed, note that \cite[Lemma 5.4]{AKLMW24} and the isometric property 
\eqref{eq:isometric_property} 
of the stochastic integral 
ensure that 
\begin{equation*}\label{eq:isometry_on_nabla_u'-nabla_u'_aux}
\begin{aligned}
& \E\Bigl[ \bigl\| \delta\;\!\nabla_\x \u'(\x,t) -
 \delta\;\!\nabla_\x \u'_{\mathrm{aux}}(\x,t) 
 \bigr\|^2 \Bigr]
\\ & \qquad  = 
\delta^2 \int_{\R^+} \int_{S^2} \int_{ \R} \, \biggl| \nabla_\x \biggl( \Bigl( \exp\Bigl\{ i \frac{1}{\delta} \;\! \kappa \;\! \btheta \cdot \flow(s;\x,t)\Bigr\} 
- \exp\Bigl\{ i \frac{1}{\delta} \;\! \kappa \;\! \btheta \cdot \bigl( \x - (t-s) \overline{\u}(\x,t)\bigr) \Bigr\} \Bigr) 
\\ & \qquad \phantom{ = 2 \su(\x,t)^2 \int_{\R^+ \times S^2 \times \R}} \;\;  \etabb(t-s;\x,t) \, \ebb(\kappa;\x,t) \, \Pbb(\btheta;\x,t)  \biggr) \biggr|^2 \d s\, U_{S^2}(\d\btheta)\, \d \kappa.
\end{aligned}
\end{equation*}
In this case we additionally introduce a truncated version
of the energy spectrum 
by defining, for any $C > 0$, $E_C(\kappa;\zeta) = \ind_{[0,C]}(\kappa) E(\kappa;\zeta)$ and $\ebb_C(\kappa;\x,t) = \sx^{1/2}(\x,t) E_C(\sx(\x,t) \kappa;\sz(\x,t) z)$, which has the property that $\ebb_C(\kappa;\x,t) \neq 0$ implies $\kappa \leq C/\sx(\x,t)$. The convergence is then established by employing the decomposition $\ebb = \ebb_C+(\ebb - \ebb_C)$ and proceeding in a similar way as before.

Part b) is an immediate consequence of Corollary \ref{cor:stratMC_randomfields}, 
using  
$\Delta_j = \R^+ \times S^2 \times [j \Delta s, (j+1)\Delta s)$, 
$p_j(\kappa,\btheta,s) = p(\kappa)\, \ind_{\Delta_j} (\kappa,\btheta,s) / \Delta s$ 
for $j \in \mathbb{Z}$, and
\begin{align*}
\bm{G}(\x,t,\kappa,\btheta,s) = 
\etabb(t-s;\x,t) 
\exp\Bigl\{ i \frac{1}{\delta}\;\! \kappa\;\! \btheta \cdot \bigl(\x-(t-s)\overline\u(\x,t)\bigr) \Bigr\} \,\ebb(\kappa;\x,t) \,
\Pbb(\btheta;\x,t)  
\end{align*}
in the notation therein.
\end{proof}

As a 
consequence 
of Theorem~\ref{thm:conv}, we obtain that the characteristic flow properties 
established 
in 
\cite[Theorem~5.5]{AKLMW24} 
for the continuous model carry over to the discretized fields $\u'_N$.

\begin{corollary}[Characteristic flow properties]\label{prop:flow_properties}
Let
$\u'_N=(\u'_N(\x,t))_{(\x,t)\in\R^3\times \R}$, $N\in\N$, 
be the family of approximating random fields specified in  Discretization~\ref{discretization}. 
Then, considering the turbulence scale ratio 
$\delta$,  
for every $(\x,t)\in\R^3\times\R$, $N\in\N$ we have that  
\begin{align}\label{eq:kinetic_energy}
\frac{1}{2}\,&\E \Bigl[ \bigl\| \u'_N(\x,t) \bigr\|^2 \Bigr]  \,=\, k(\x,t), 
\\ \label{eq:dissipation}
\lim_{\delta \to 0} \,
\frac12 \,\delta ^2 z  \,&\E \Bigl[ \bigl\| \nabla_{\x} \u'_N(\x,t) + \bigl( \nabla_{\x} \u'_N(\x,t)\bigr){\vphantom{)}}^{\!\top} \bigr\|^2 \Bigr]  \,=\, \frac{\eps(\x,t)}{\nu(\x,t)}, 
\\ \label{eq:divergence}
\lim_{\delta \to 0}  \;
&\E \Bigl[ \bigl| \delta \, \nabla_{\x} \cdot\u'_N(\x,t) \bigr|^2 \Bigr]  \,=\, 0,
\end{align} 
where the 
first equality holds 
for all values of 
$\delta$. 
In addition and consistent with \eqref{eq:kinetic_energy},  
the one-point velocity correlations satisfy 
\begin{equation}\label{eq:one-point_correlation}
\begin{aligned}
\E\bigl[\u'_N(\x,t)\otimes\u'_N(\x,t)\bigr] 
= k(\x,t)\Bigl[\frac{7}{15}\,\bm{L}(\x,t)\cdot\bm{L}(\x,t){\vphantom{)}}^{\!\top} + \frac{1}{5}\,\bm{I}\Bigr]. 
\end{aligned} 
\end{equation}
\end{corollary}

\begin{proof}
Observe that \cite[Lemma A.8]{AKLMW24} guarantees that the expected values 
appearing in \eqref{eq:kinetic_energy}--\eqref{eq:one-point_correlation} 
are fully determined by 
the covariance structures of the  
considered random fields. 
The fact that the covariance structures of 
$\u'_{\mathrm{aux}}$ and $\u_N$ coincide according to Theorem~\ref{thm:conv}    
therefore implies that it is sufficient to 
verify the assertions \eqref{eq:kinetic_energy}--\eqref{eq:one-point_correlation} for $\u'_{\mathrm{aux}}$ in place of $\u_N$. 
This in turn is achieved by reasoning directly along the lines of the proof of \cite[Theorem 5.5]{AKLMW24}, 
replacing the mean flow function $\flow(s;\x,t)$ appearing therein by its easier to analyze linearized counterpart $\x + (s-t)\overline{\u}(\x,t)$.
\end{proof}

\section{Algorithmic implementation}
\label{sec:implementation}

This section 
discusses 
the algorithmic implementation 
of the discretized 
inhomogeneous random field 
model 
analyzed 
in the previous section. 
The 
sampling 
procedure 
presented in Algorithm~\ref{algorithm} 
below 
is particularly flexible in that it 
allows for localized  
simulations 
of a sample path of 
the discretized field $\u'_N$ at evaluation points $(\x,t)$ that 
may be determined dynamically and need not be known in advance. 

We first 
reformulate 
the representation formula \eqref{eq:uN'2} 
for 
$\u'_N$ 
in Discretization~\ref{discretization} 
in a 
suitable 
way by 
explicitly 
writing out 
the real parts of the 
involved 
complex products. 
Further using 
the abbreviations 
$\etabb(s;\x,t)$, $\mathbbm{e}(s;\x,t)$, and $\Pbb(\btheta;\x,t)$ 
for the scaled time integration kernel, the scaled energy spectrum, and the 
involved 
transformation matrix 
given 
in  \eqref{eq:def_etabb_fxbb},  
the representation formula reads 
\begin{equation}\label{eq:u'N_implementation}
\begin{aligned}
\u'_N(\x,t) = \sum_{j\in\mathbb Z}\frac{1}{\sqrt{N}}\sum_{n=1}^N
&\, \etabb(t-s_{jn};\x,t) \,
\mathbbm{e}(\kappa_{jn};\x,t) 
\;\!\Bigl(\frac{\Delta s}{ p(\kappa_{jn}) }\Bigr)^{1/2} 
\\ & 
\,\Pbb(\btheta_{jn};\x,t)\cdot 
\bigl[\cos(\alpha_{jn})\, \Re\;\!\wn_{jn} - \sin(\alpha_{jn})\, \Im\;\!\wn_{jn}\bigr],
\end{aligned}
\end{equation}
where we have 
additionally 
introduced 
the notation 
\begin{align*}
\alpha_{jn}
&= 
\frac1\delta \kappa_{jn}\btheta_{jn} \cdot \bigl( \x-(t-s_{jn})\overline\u(\x,t) \bigr).
\end{align*}

The random quadrature points $(\kappa_{jn},\btheta_{jn},s_{jn})$ 
appearing 
in \eqref{eq:u'N_implementation} 
can be sampled by standard simulation techniques \cite{MNR12,RK17}. 
Specifically, 
we 
employ
the inverse transformation method 
in order 
to sample the  
wave numbers $\kappa_{jn}$ 
according the probability density function 
$p(\kappa)$,
while the 
uniformly  distributed 
orientation vectors $\btheta_{jn}$ on the unit sphere $S^2$ are obtained 
via radial projection of 
standard normal random vectors in $\R^3$. 
The noise vectors $\Re\;\!\wn_{jn}$, $\Im\;\!\wn_{jn}$, 
whose distribution 
is not fully specified in Discretization~\ref{discretization}, are 
chosen 
to be standard 
normally distributed
as well. 
Note that the covariance structure of the noise vectors is adjusted 
via the anisotropy factor $\bm{L}(\x,t)$  
occurring in $\Pbb(\btheta_{jn};\x,t)$.
In order to 
achieve 
prescribed directional weightings for the  one-point velocity correlations (Reynolds stresses), 
this factor can be calculated by means of a Cholesky decomposition of 
an auxiliary matrix related to the 
Reynolds stress tensor; compare \cite[Theorem~5.5]{AKLMW24} and  Corollary~\ref{prop:flow_properties} above. 
In the case of isotropic one-point velocity correlations
the factor simplifies to the identity matrix $\bm{L}(\x,t)=\bm{I}$. 

From an application point of view, it is desirable to be able to 
efficiently 
simulate 
a sample path of $\u'_N$ 
at 
selected 
evaluation points $(\x,t)$ that may be 
determined 
successively 
in the course of the simulation process. 
As 
the temporal quadrature points $s_{jn}$ are 
located in 
the 
respective 
stratification intervals $I_j=[j\Delta s,(j+1)\Delta s)$, 
it is clear that 
for a given evaluation point $(\x,t)$ 
the 
sum $\sum_{j\in\mathbb{Z}}$ 
in \eqref{eq:u'N_implementation} can be restricted 
to those indices $j$ which represent 
intervals
that overlap with the 
support 
of the weight function $s \mapsto \etabb(t-s;\x,t)$.
This support is bounded and 
depends on both the current time 
$t$ and the spatial evaluation point $\x$,  
as   
its length 
is
adjusted in terms of the temporal scaling factor $\st(\x,t)=k(\x,t)/\eps(\x,t)$ 
in \eqref{eq:def_etabb_fxbb}.
In order to ensure consistency of evaluations of a sample path of $\u'_N$ at two distinct 
spatio-temporal 
evaluation points, 
it is necessary to use identical samples  values $\kappa_{jn}$, $\btheta_{jn}$, $s_{j,n}$, $\wn_{j,n}$ 
for both evaluations if 
a
stratification interval $I_j$ contributes to both of them.  
It is therefore crucial to systematically keep track of the stratification intervals and random numbers 
employed in the simulation.   

Algorithm~\ref{algorithm} below describes 
a procedure for generating samples of 
$\u'_N$ at a given time point $t$ and one or multiple spatial evaluation points $\x$ 
in such a way that consistency with possible previous evaluations of the sample path of $\u'_N$ at previous time points 
$t_\text{old}\le t$ 
in terms of  the employed random numbers
is guaranteed.  
For definiteness we assume that the support of the time integration kernel $\eta(s)$ 
is 
given by 
the compact 
interval with boundary points $\pm C$, where $C\in\R^+$. 
The support of the 
scaled 
kernel 
$s \mapsto \etabb(t-s;\x,t)$ 
thus 
coincides with the  
compact 
interval with boundary points $t\pm \delta\;\!\st(\x,t) C$. 
In order to indicate the 
range of stratification intervals $I_j$ contributing to the 
evaluations 
of $\u'_N$ 
at the current time point $t$,  
we employ 
the indices 
\begin{equation} \label{eq:ind_lowerupper}
\verb!ind_lower! = \big\lfloor \bigl(t - \delta \max_{\x}\st(\x,t) \;\! C\bigr) / \Delta s \big\rfloor, 
\quad 
\verb!ind_upper!  = \big\lceil \bigl(t + \delta \max_{\x}\st(\x,t) \;\! C\bigr) / \Delta s \big\rceil; 
\end{equation}
see Figure~\ref{fig:implementation} for an illustration. 
Here 
$\lfloor\cdot \rfloor$ and $\lceil \cdot \rceil$ denote the floor and ceiling function, respectively, and
the maximum is taken locally over all 
spatial evaluation 
points 
$\x$ where the random field $\u_N'$ has to be sampled at 
time $t$.
In addition, we use the index 
\begin{equation} \label{eq:ind_lower_glob}
\verb!ind_lower_glob! = \big\lfloor \bigl(t - \delta \, \max_{(\tilde\x,\tilde t\;\! )}\st(\tilde\x,\tilde t\;\! ) \;\! C\bigr) / \Delta s \big\rfloor,
\end{equation}
as a lower bound for those indices $j$ 
which correspond to stratification intervals $I_j$ that are 
potentially relevant for 
future 
evaluations of the sample path of $\u'_N$ at subsequent time points $t_{\text{new}}\ge t$. 
In \eqref{eq:ind_lower_glob} the maximum is taken globally over all points $(\tilde \x,\tilde t\;\!)$ in the 
spatio-temporal domain of simulation of the underlying $k$-$\eps$ model. 
In view of the necessity to employ consistent sample values 
$\kappa_{jn}$, $\btheta_{jn}$, $s_{j,n}$, $\wn_{j,n}$, 
the range of potentially relevant stratification intervals $I_j$ 
indicated by 
the lower bound in \eqref{eq:ind_lower_glob} 
and the upper index in \eqref{eq:ind_lowerupper} has to be compared with 
a corresponding range associated to the preceding evaluation time 
$t_\text{old}$ 
in order to identify those intervals $I_j$ for which previously drawn sample values 
have to reused.
The latter range is 
specified 
by suitable indices  $\verb!ind_lower_saved!$ and $\verb!ind_upper_saved!$, compare Figure~\ref{fig:implementation}. 
 
\begin{figure}
  \centerline{\includegraphics[trim={3cm 3.5cm 2cm 3.3cm}, clip, scale=0.8]{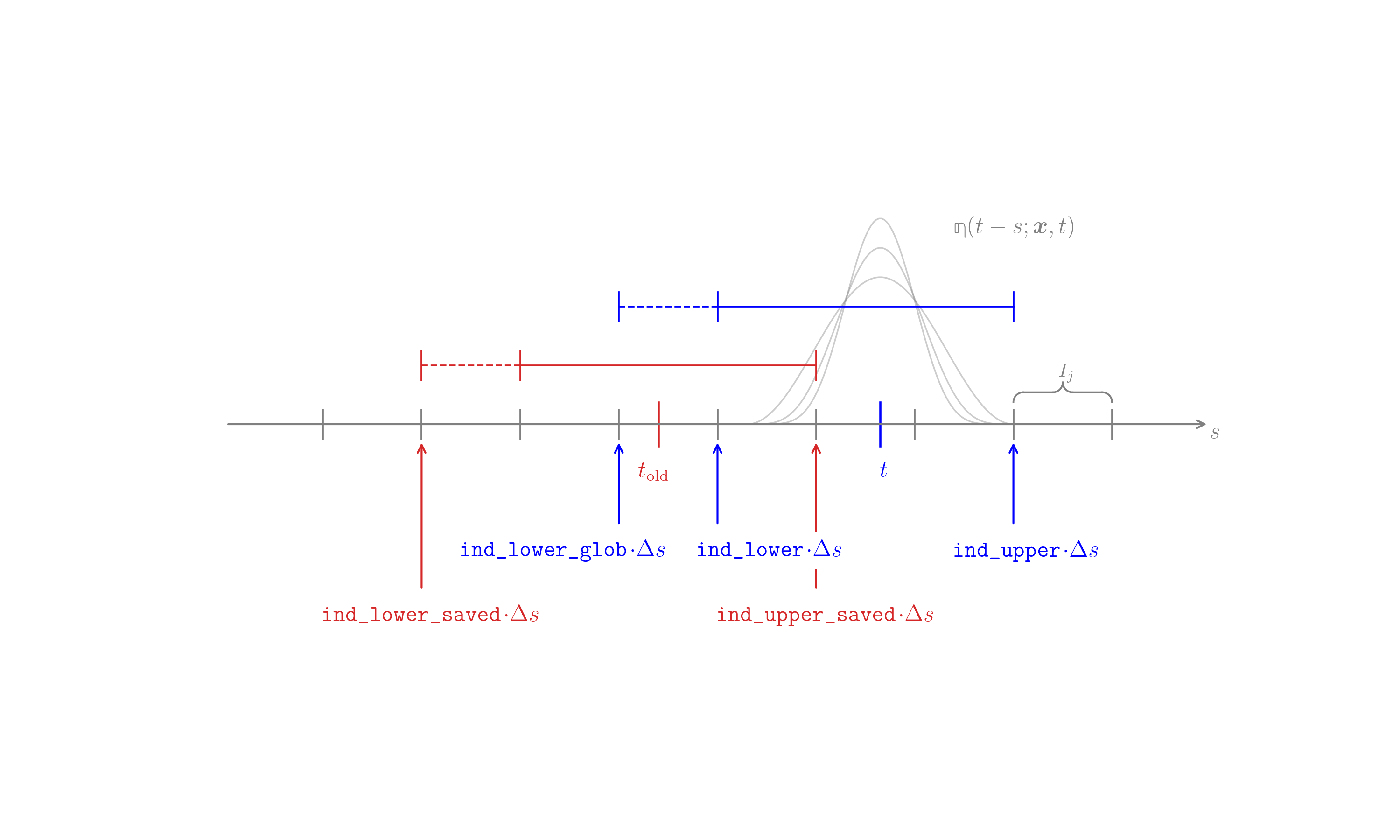}}
  \caption{
  Illustration of the indices used in Algorithm~\ref{algorithm}. 
  The sample path of $\u'_N$ is 
  evaluated at time $t$ and one or multiple spatial evaluation points $\x$, 
  specifying one or multiple scaled time integration kernels $s\mapsto \etabb(t-s;\x,t)$ 
  with possibly different supports. 
  The range indicated by the 
  horizontal 
  solid blue line comprises those intervals $I_j$ that contribute to the evaluations at time $t$. See the text for details.
  }
  \label{fig:implementation}
\end{figure}

For the formulation of Algorithm~\ref{algorithm} we assume that 
the numerical parameters $N$ and $\Delta s$ have been 
chosen 
and that the global paramters $\delta$, $z$, and $\max_{(\tilde\x,\tilde t\;\! )}\st(\tilde\x,\tilde t\;\! )$ are available. 

\begin{algorithm}[Sampling procedure] \label{algorithm}
\phantom{.} \vspace{3mm}\\
\textbf{output}\,:  \;
\begin{minipage}[t]{0.9\textwidth}
sample of $\u'_N$ at time $t$ and one or multiple spatial evaluation points $\x$,  
consistent with \\
possible previous evaluations of the sample path of $\u'_N$ at previous time points 
in terms \\ of  the employed random numbers
\end{minipage}
\vspace{2.3mm}\\
\textbf{input}\,:  \hspace*{-4.2mm}
\begin{minipage}[t]{0.92\textwidth}
\begin{itemize}
\item 
evaluation time $t$ and 
one or multiple spatial evaluation points $\x$
\item 
flow data $\overline\u$, $\nu$, $k$, $\eps$ 
and Reynolds stress tensor $\bm{R}$ 
at evaluation points 
$(\x,t)$
\item
range of 
indices 
$j = $ $\verb!ind_lower_saved!,$  $\ldots ,  \verb!ind_upper_saved!-1$
and associated \\
sample values $s_{jn}$, $\kappa_{jn}$, $\btheta_{jn}$, $\Re\;\!\wn_{jn}$, $\Im\;\!\wn_{jn}$ 
saved from 
possible previous evaluations of  \\
the sample path of $\u'_N$ at prededing time point
\end{itemize}
\end{minipage}
\vspace{4mm}\\
determine indices \verb!ind_lower!, \verb!ind_upper! and \verb!ind_lower_glob! according to \eqref{eq:ind_lowerupper} and \eqref{eq:ind_lower_glob} 
\vspace{4mm}\\
\textbf{if} \;
\begin{minipage}[t]{0.8\textwidth}
sample path of $\u'_N$ has not been evaluated previously\,:  
\end{minipage}
\vspace{3.5mm}\\
\hspace*{10.5mm}
\begin{minipage}[t]{0.9\textwidth}
generate and save sample values $s_{jn}$, $\kappa_{jn}$, $\btheta_{jn}$, $\Re\;\!\wn_{jn}$, $\Im\;\!\wn_{jn}$ 
corresponding to indices \\
$j = \verb!ind_lower_glob! , \ldots ,  \verb!ind_upper!-1$, \,\! $n=1,\ldots, N$
\vspace{3.5mm}\\
save indices
$\verb!ind_lower_saved! \leftarrow \verb!ind_lower_glob!$, 
$\verb!ind_upper_saved! \leftarrow \verb!ind_upper!$
\end{minipage}
\vspace{4mm}\\
\textbf{else}\,: 
\vspace{2mm}\\
\hspace*{10.5mm}
\begin{minipage}[t]{0.9\textwidth}
delete sample values $s_{jn}$, $\kappa_{jn}$, $\btheta_{jn}$, $\Re\;\!\wn_{jn}$, $\Im\;\!\wn_{jn}$ 
with
indices $j < \verb!ind_lower_glob!$ \\
and save index $\verb!ind_lower_saved! \leftarrow \verb!ind_lower_glob!$
\vspace{3.5mm}\\
\textbf{if} \;
\begin{minipage}[t]{0.8\textwidth}
$\verb!ind_upper_saved! < \verb!ind_upper!$\,:
\end{minipage}
\vspace{3.5mm}\\
\hspace*{10.5mm}
\begin{minipage}[t]{0.93\textwidth}
generate and save sample values $s_{jn}$, $\kappa_{jn}$, $\btheta_{jn}$, $\Re\;\!\wn_{jn}$, $\Im\;\!\wn_{jn}$ 
corresponding to 
indices \\
$j = \max\bigl(\verb!ind_lower_glob!\,,\;\! \verb!ind_upper_saved!\bigr) , \ldots ,  \verb!ind_upper!-1$, \,\! $n=1,\ldots, N$
\vspace{3.5mm}\\
save index $\verb!ind_upper_saved! \leftarrow \verb!ind_upper!$
\end{minipage}
\end{minipage}
\vspace{4mm}\\
calculate sample of $\u'_N$ at evaluation points $(\x,t)$ according to \eqref{eq:u'N_implementation} 
with summation index $j$\\ 
 ranging from $\verb!ind_lower!$ to $\verb!ind_upper!-1$ 
 \vspace{3.5mm}\\
 return sample values $\u'_N(\x,t)$ 
 \vspace{1.5mm}
\end{algorithm}

We 
end 
this section by 
noting 
that a suitable choice for the stratification length $\Delta s$ 
is half the length of the smallest possible support of 
the scaled time integration kernel 
$s \mapsto \etabb(t-s;\x,t)$,    
i.e., 
$ 
\Delta s = \delta \min_{(\tilde \x,\tilde t\;\!)}\st(\tilde \x,\tilde t\;\!)\;\!  C,
$
where the minimum is taken globally over all points $(\tilde \x,\tilde t\;\!)$ in the relevant spatio-temporal domain of simulation of the underlying $k$-$\eps$ model. 
In dependence on the number $N$ of quadrature points per stratification interval, 
this ensures an adequate 
control of the  
minimum approximation quality.

\section{Simulation results and discussion of the model}\label{sec:sim_res}

Here we present 
various 
numerical simulation results illustrating the specific features of 
our inhomogenous random field model. 
Subsection~\ref{sec:parameters} addresses the influence of the model parameters
on the generated 
fluctuations,  
with 
particular emphasis on the 
$(\x,t)$-dependent 
scaling factors 
and the 
inhomogeneous 
mean flow function.  
In Subsection~\ref{sec:ergodicity} we demonstrate the ergodicity properties of the model 
by recovering the underyling flow fields of kinetic turbulent energy 
$k$ 
and dissipation rate 
$\eps$ 
in terms of local 
sample path 
averages in space and time. 
Subsection~\ref{sec:Kolmogorov} 
concerns 
the reproducibility of Kolmogorov's two-thirds law for the second moments of the spatial velocity increments
 in the inertial subrange. 
To facilitate 
the 
exposition, 
the simulations are set up in 
stylized scenarios 
that allow to 
emphasize  
different 
aspects 
separately from each other.

Throughout this section 
we close Model~\ref{model:turb_inhom} by employing  
the 
spatial energy 
spectrum 
$E$ 
and the time integration kernel 
$\eta$  
specified in \cite[Examples 2.3 and 2.4]{AKLMW24}. 
All 
presented 
simulations are based on 
the numerical approximation scheme 
in Discretization~\ref{discretization}   
and its algorithmic implementation as 
described 
in Section~\ref{sec:implementation} above. 
The 
reference density 
for the random wave numbers 
is 
taken as 
$p(\kappa)=E(\kappa;z)$. 
The 
stratification length $\Delta s$ is 
chosen as  
half the length of the smallest possible support of 
the scaled time integration kernel 
$s \mapsto \etabb(t-s;\x,t)$
in each considered scenario.

\subsection{Significance of the model parameters} 
\label{sec:parameters} 

In this subsection we 
illustrate  
the influence of the 
parameters 
involved 
in the definition of 
the inhomogeneous turbulence field 
in Model~\ref{model:turb_inhom}.
The 
main focus lies on the 
$(\x,t)$-dependent 
scaling factors 
$\sx=k^{3/2}/\eps$, $\st=k/\eps$,  
$\su=k^{1/2}$, 
and $\sz=\eps\nu/k^2$, 
specifying 
the
turbulence scales 
and 
the 
inverse 
turbulent viscosity ratio 
(inverse turbulence Reynolds number) 
prescribed by the flow fields 
$k$, $\eps$, and $\nu$. 
Different scenarios for the flow fields are considered, 
each of which  
highlights 
one of the scaling factors varying 
in space or time,   
while the other factors are kept constant 
unless interdependencies imply otherwise.  
We also discuss the influence of the mean flow function $\flow$ 
defined in \eqref{eq:mean_flow},  
modeling 
the advection of the turbulent structures along the mean velocity field $\overline \u$. 
In order to simplify the interpretation of the 
simulation results, 
the 
anisotropy factor   
$\bm{L}$ 
is chosen as 
the identity matrix 
throughout,   
$\bm{L}(\x,t)=\bm{I}$, 
yielding 
isotropic one-point velocity correlations.  

We begin by 
presenting 
streamline plots based on a two-dimensional 
variant 
of our 
model, 
which permits 
to illustrate some of the features of the three-dimensional model in a simplified way.  
To this end, we define an 
$\R^2$-valued 
fluctuation field $\u'=(\u'(\x,t))_{(\x,t)\in\R^2\times \R}$ in analogy to Model~\ref{model:turb_inhom}, 
replacing the unit sphere $S^2$ in the representation formula \eqref{eq:u'} by the unit circle $S^1$ and employing an underlying white noise with values in $\C^2$ instead of $\C^3$. 
We remark that this two-dimensional analogue is used merely for 
illustration purposes
and 
is not intended to 
capture 
the well-known 
structural differences between three-dimensional 
and two-dimensional 
turbulence 
\cite[Chapter~10]{Dav15}, 
which lie beyond the scope of this article. 
In particular,  
we 
employ 
the energy spectrum from  \cite[Example~2.3]{AKLMW24},   
obeying  
Kolmogorov's $5/3$ law for three-dimensional turbulence. 
Approximate realizations 
of the two-dimensional fluctuation field at a fixed time point $t$ are shown in 
Figures~\ref{fig:2D_sx}--\ref{fig:2D_su} 
in the form of streamlines on 
a rectangular domain.  
Here the  
turbulence scale ratio is 
chosen as 
$\delta = 0.08$, 
the mean velocity field $\overline \u$ is assumed to be identically zero, and the numerical parameter 
governing the number of quadrature points
is taken 
as 
$N=4\,000$. 


\begin{figure}[h!]
	\centerline{\includegraphics[width=.8\textwidth]{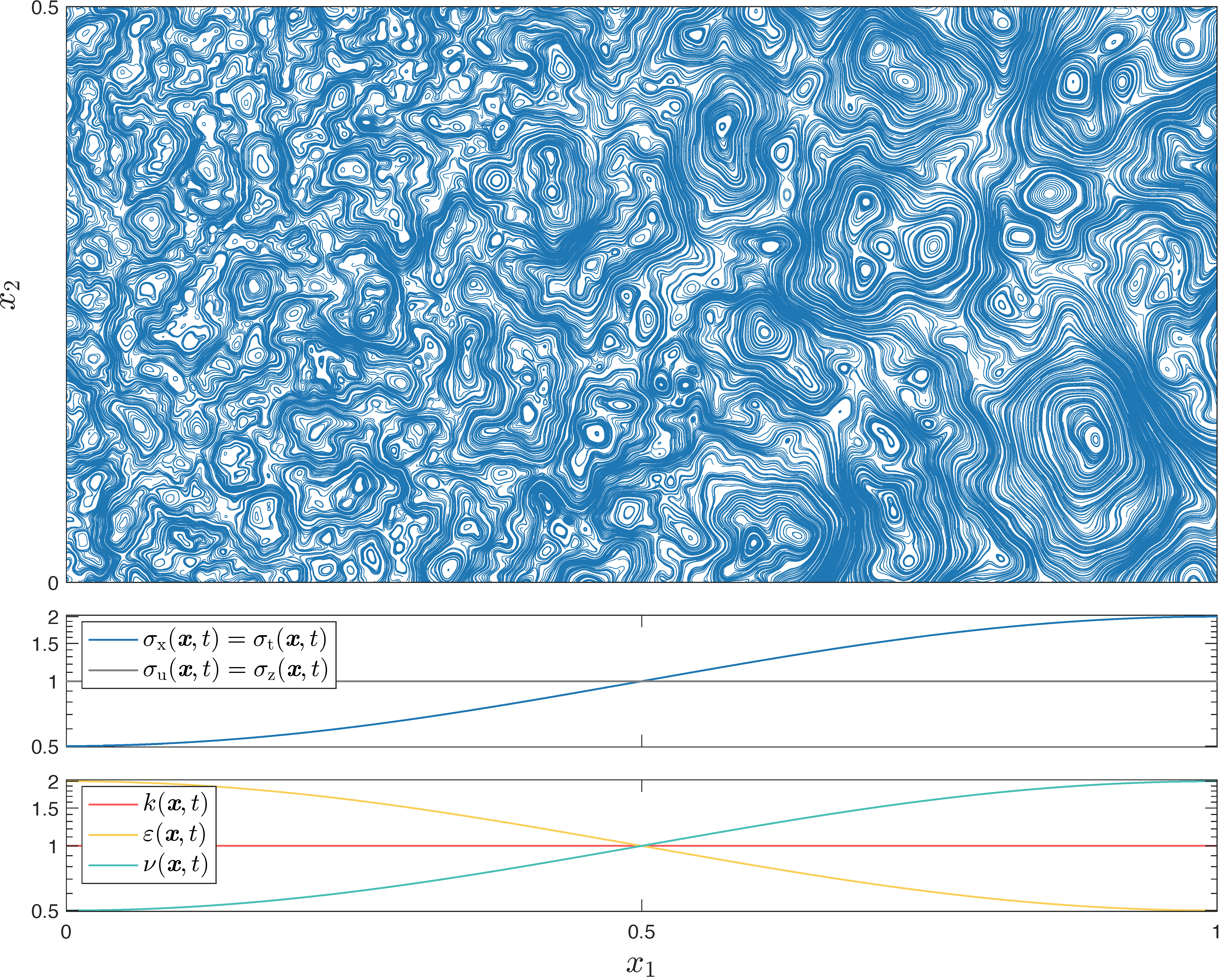}}
	\caption{Streamline plot showing an approximate realization of a two-dimensional analogue of Model~\ref{model:turb_inhom} at a fixed time point,   
	with spatial scaling 
	function 
	$\sx=k^{3/2}/\eps$ 
	increasing along the $x_1$-axis  
	by a factor of four. 
	The scaling functions and underlying flow quantities are depicted in semi-log plots. 
	See the text for details.
	} 
	 \label{fig:2D_sx}	
\end{figure}

Figure~\ref{fig:2D_sx} presents a scenario in which the spatial scaling 
function
$\sx$ 
increases along the $x_1$-axis by a factor of 
four 
and does not depend on 
the $x_2$-coordinate, 
whereas 
$\su$ 
and 
$\sz$ 
are constant with value one.  
Specifically, $\sx$ has an $S$-shaped graph 
on a logarithmic scale 
and is given by 
$\sx(\x,t) = 2^{\;\!\cos(\pi + \pi x_1)}$   
for $\x=(x_1,x_2)\in [0,1]\times[0,0.5]$. 
Accordingly, the 
flow fields of kinetic turbulent energy, dissipation rate, and kinematic viscosity 
are chosen as 
$k(\x,t)=1$, $\eps(\x,t)=2^{\;\!-\!\cos(\pi + \pi x_1)}$, and $\nu(\x,t)=2^{\;\!\cos(\pi + \pi x_1)}$. 
The 
characteristic value for the 
inverse turbulent viscosity ratio is  
set to 
$z=0.005$, 
and we further note that  
$k$ being 
normalized 
entails the identity $\st=\sx$. 
In 
agreement 
with the behavior of the 
spatial scaling factor 
$\sx$, 
it can be seen that 
the turbulent structures grow in size 
from left to right, with a 
scale 
ratio of roughly four 
in regard to 
the edge regions 
near $x_1=0$ and $x_1=1$. 
Apart from the 
difference in scale,  
the composition of the  
structures 
in terms of 
superposed vortices of different sizes 
appears to be similar in all regions, 
reflecting the fact that the 
factor $\sz$ 
governing the shape of the energy spectrum
is held constant.

\begin{figure}[h!]
	\centerline{\includegraphics[width=.8\textwidth]{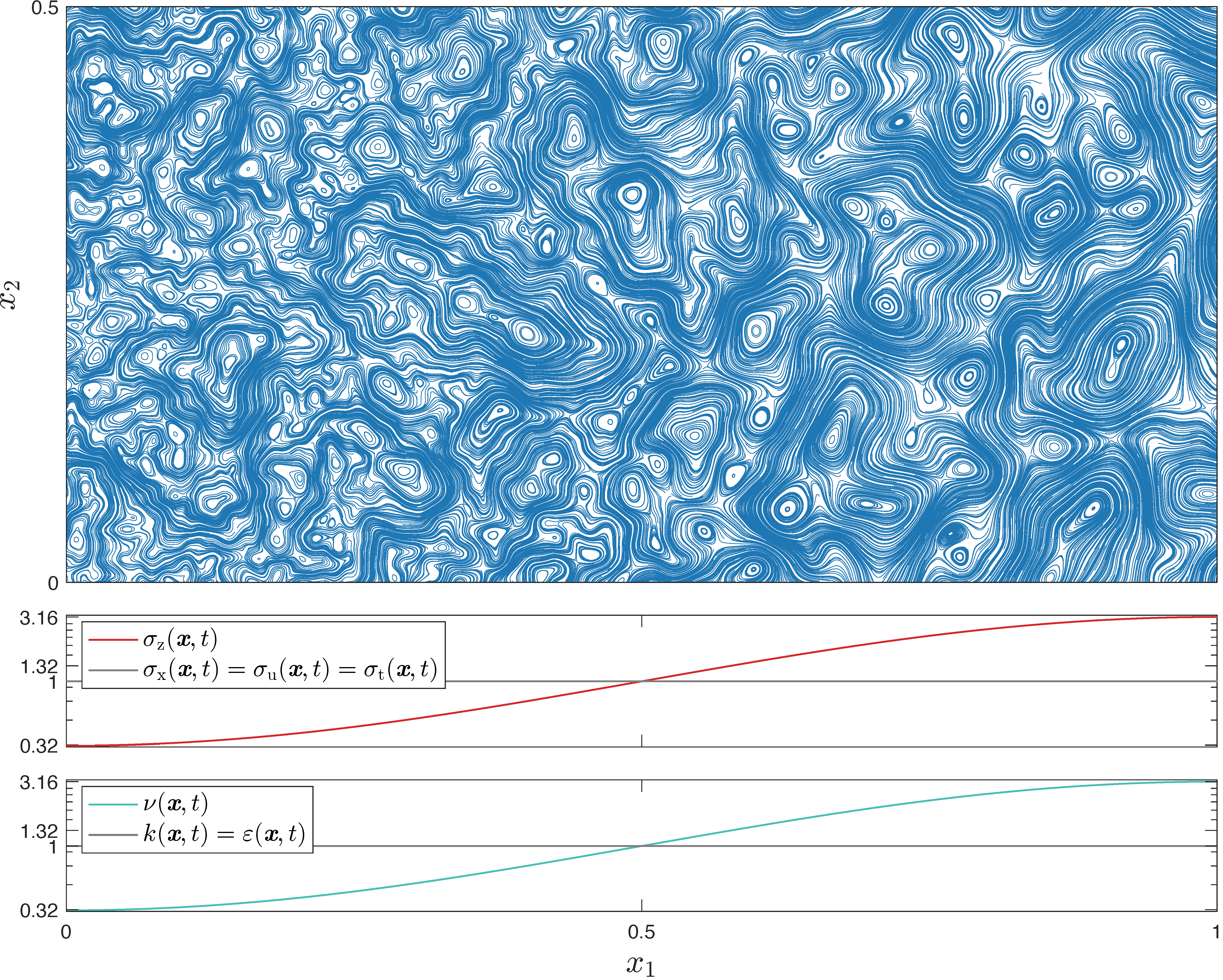}}
	\caption{Streamline plot showing an approximate realization of a two-dimensional analogue of Model~\ref{model:turb_inhom} at a fixed time point,  
	with 
	viscosity 
	scaling 
	function 
	$\sz=\eps\nu/k^2$ 
	increasing along the $x_1$-axis 
	by a factor of ten. 
	The scaling functions and underlying flow quantities are depicted in semi-log plots.  
	See the text for details.
	} 
	 \label{fig:2D_sz}	
\end{figure}

In contrast, the scenario in Figure~\ref{fig:2D_sz}  involves 
an increase 
of $\sz$ along the $x_1$-axis by a factor of ten, while $\sx$, $\su$, and $\st$ are constant with value one. 
Here we assume
$\sz(\x,t)=\nu(\x,t)=10^{\;\!\cos(\pi+\pi x_1)/2}$ 
and $k(\x,t)=\eps(\x,t)=1$ 
as well as 
$z=0.01$.   
Recalling that the factor 
$\sz=\eps\nu/k^2$ 
is associated to 
the inverse of the turbulent viscosity ratio 
and 
the inverse of the turbulence Reynolds number, 
we note that 
the 
$3/4$\;\!th 
power of 
$\sz z=(\eps\nu/k^2)(\eps_0\nu_0/k_0^2)$ 
represents the scale ratio between the turbulent fine-scale (Kolmogorov) and large-scale structures  
\cite[Section~6.3]{Pope00}. 
Moreover, in
the employed model spectrum from \cite[Example~2.3]{AKLMW24} the parameter  
$\zeta = \sz z$
implicitly 
determines the transition wave numbers 
$0<\kappa_1(\zeta) < \kappa_2(\zeta)<\infty$ 
indicating the  
inertial subrange, 
whose width increases as 
$\zeta$ 
decreases. 
Specifically, 
the transition 
wave number 
$\kappa_1(\zeta)$ 
associated to the turbulent large-scale structures 
shows 
only a minor dependence on $\zeta$, 
while 
the transition wave number 
$\kappa_2(\zeta)$ 
associated to the turbulent small-scale structures 
tends 
to infinity as $\zeta\to 0$. 
In accordance with 
these interrelations 
and the fact that the spatial scaling factor $\sx$ is held constant, 
it can be observed that 
the size of 
the large-scale structures in the streamline plot 
in Figure~\ref{fig:2D_sz} 
is roughly uniform 
throughout the 
rectangular domain.  
In addition, 
the large-scale vortices are superposed by smaller 
vortices, 
and  
it is visible 
that both the rate and the range of 
the smaller vortex sizes 
decrease as 
the value of the scaling factor $\sz$ 
increases.  
The 
composition of the 
turbulent structures on the left-hand side 
can thus be attributed 
to local energy spectra 
involving a wider variety of different wavenumbers 
than the comparably narrow spectra 
corresponding to 
the structures on the right-hand side.

\begin{figure}[h!]	
	\centerline{\includegraphics[width=.8\textwidth]{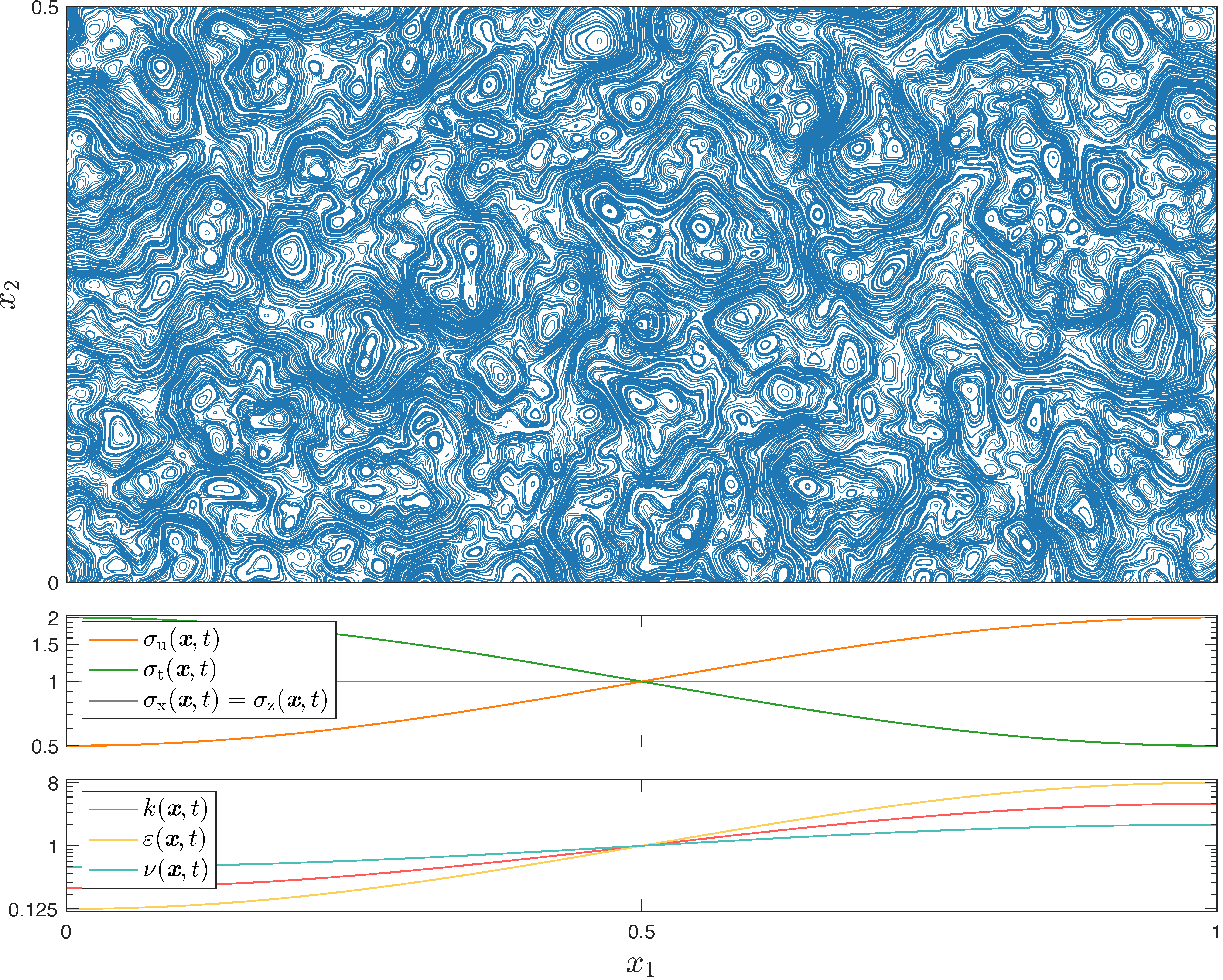}}
	\caption{Streamline plot showing an approximate realization of a two-dimensional analogue of Model~\ref{model:turb_inhom} at a fixed time point, 
	with velocity scaling 
	function 
	$\su=k^{1/2}$ 
	increasing along the $x_1$-axis  
	by a factor of four. 
	The scaling functions and underlying flow quantities are depicted in semi-log plots.   
	See the text for details.
	} 
	\label{fig:2D_su}
\end{figure}

The final streamline plot in Figure~\ref{fig:2D_su} 
focuses on the velocity scaling function $\su$,  
which is assumed to increase 
from left to right 
by a factor of four, 
while $\sx$ and $\su$ are held constant with value one. 
The scenario is specified by 
$\su(\x,t)=2^{\;\!\cos(\pi+\pi x_1)}$ 
with flow fields 
$\nu(\x,t)=2^{\;\!\cos(\pi+\pi x_1)}$, 
$k(\x,t)=4^{\;\!\cos(\pi+\pi x_1)}$,
$\eps(\x,t)=8^{\;\!\cos(\pi+\pi x_1)}$ 
and characteristic number $z=0.005$, 
further implying the identity $\st=\su^{-1}$. 
Unlike in the previous plots, 
the composition of the turbulent structures 
does not 
show distinct qualitative changes 
in dependence on the location along the $x_1$-axis. 
This is consistent with the fact that 
the only non-constant scaling factors are $\su$ and $\st$.  
While $\su$ affects the length of the velocity vectors $\u'(\x,t)$, it does not alter their direction  
and therefore only influences the length of the streamlines but not their shape. 
The difference in length of the streamlines is 
not noticeable due to their thinness and overlapping behavior. 
For the temporal scaling factor $\st$ no influence on 
the statistical properties of the fluctuation field at a fixed time point 
was to be expected 
other than 
minor 
effects related to the numerical approximation.

\begin{figure}[h!]
	\centerline{\includegraphics[width=.7\textwidth]{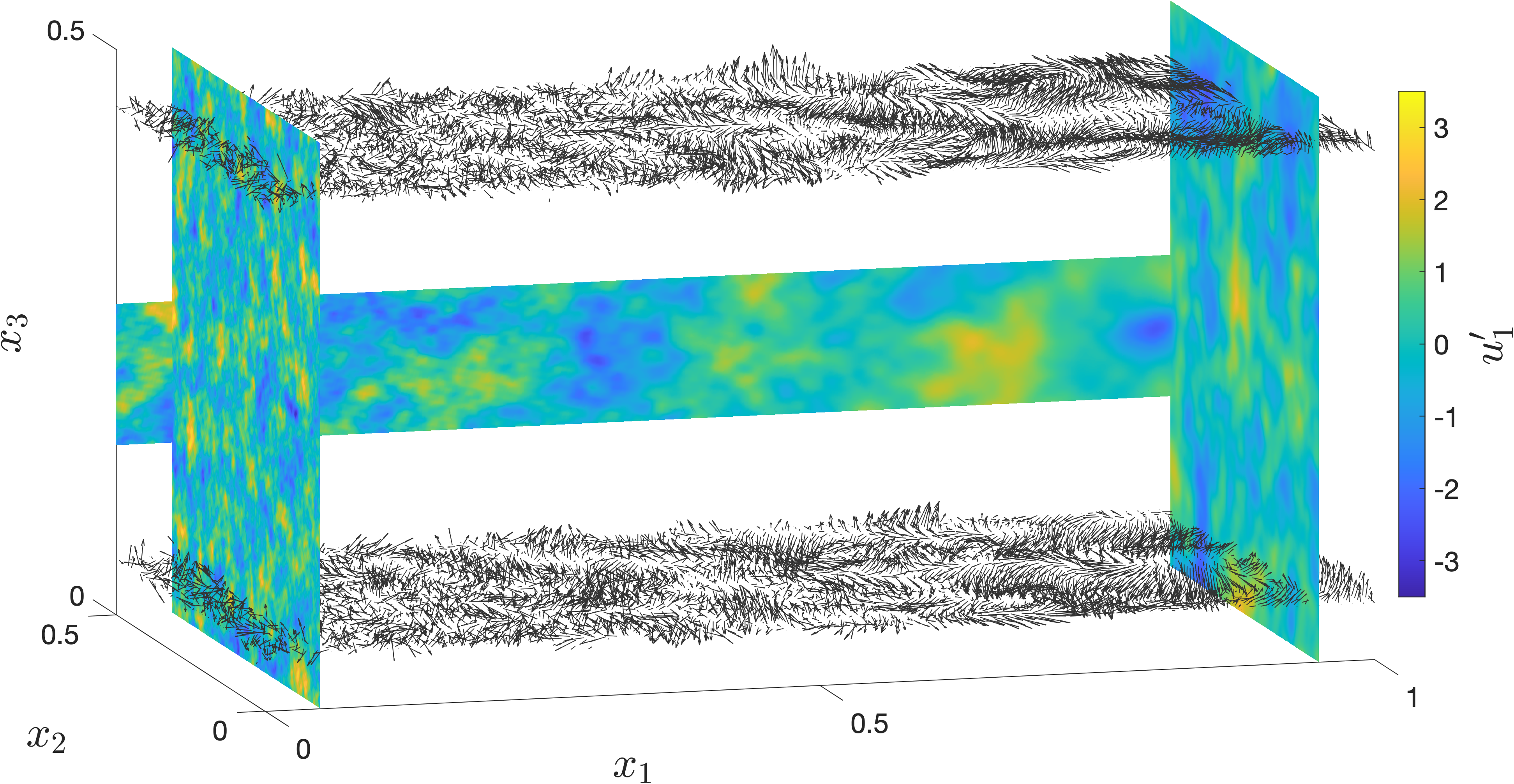}}\vfill
	\centerline{\includegraphics[width=.7\textwidth]{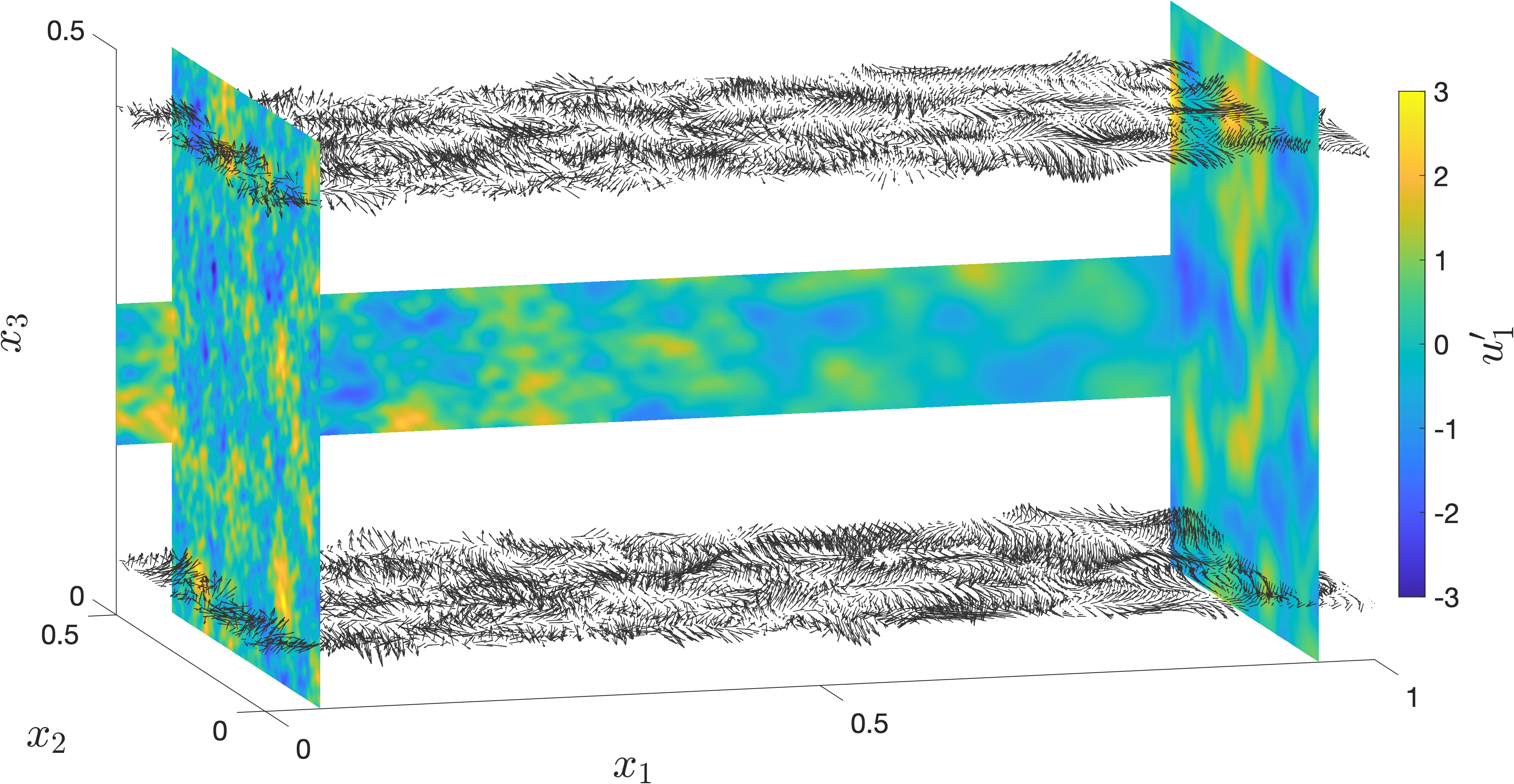}}\vfill
	\centerline{\includegraphics[width=.7\textwidth]{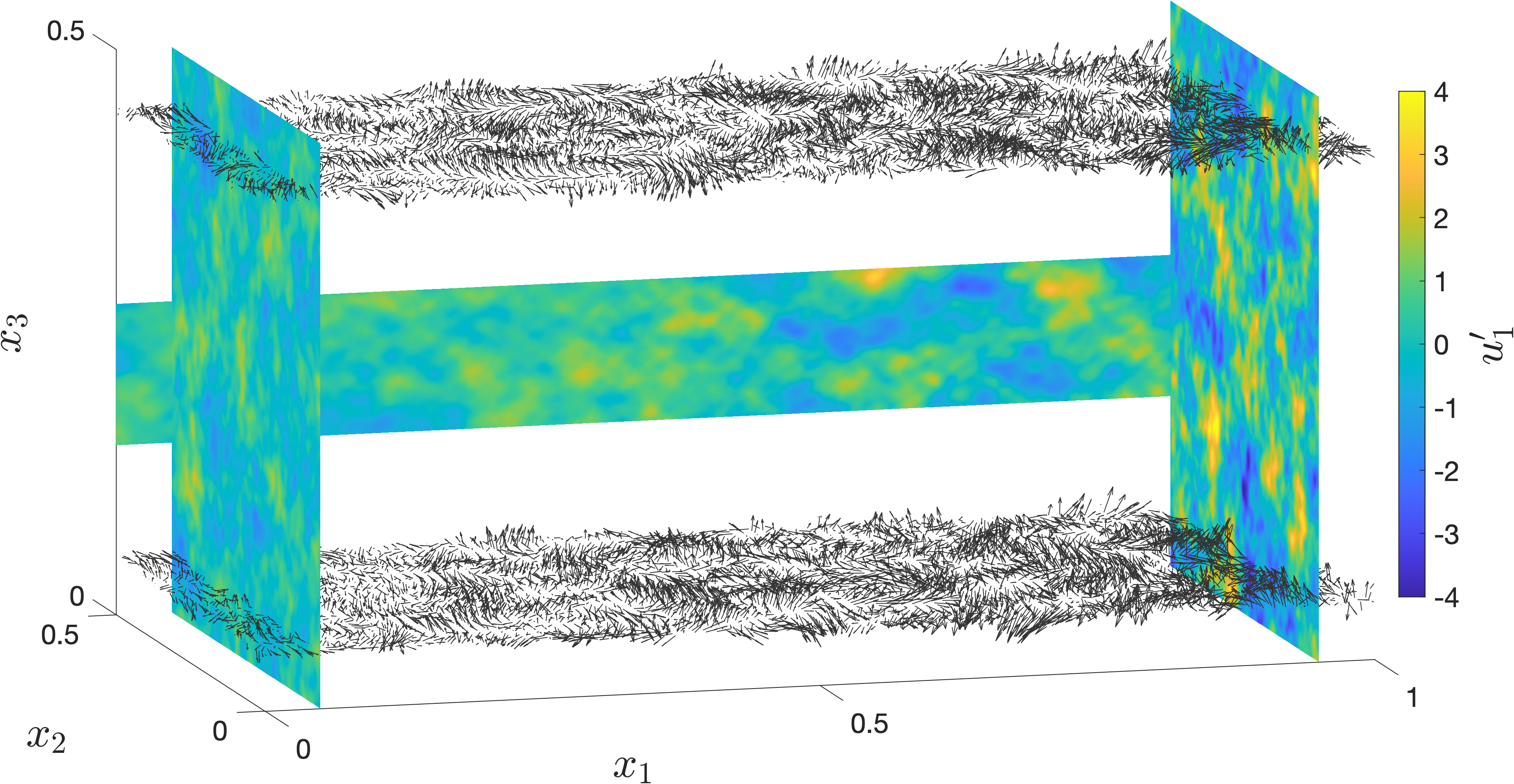}}\vfill
	\caption{ 
	Approximate realizations of the turbulence field 
	in Model~\ref{model:turb_inhom} 
	at a fixed time point. 
	The heat maps show the first velocity component,  
	and the plotted velocity vectors 
	are scaled by a factor of 
	$0.01$.
	Top: 
	Spatial scaling function $\sx=k^{3/2}/\eps$ 
	increasing along the $x_1$-axis by a factor of four 
	(scenario from Figure~\ref{fig:2D_sx}).
	Middle: 
	Viscosity scaling function $\sz=\eps\nu/k^2$ 
	increasing along the $x_1$-axis by a factor of ten 
	(scenario from Figure~\ref{fig:2D_sz}).
	Bottom:  
	Velocity scaling function $\su=k^{1/2}$ 
	increasing along the $x_1$-axis by a factor of four 
	(scenario from Figure~\ref{fig:2D_su}).   
	See the text for details.
	}
	\label{fig:3D_sxzu}
\end{figure}
 
The scenarios considered so far in the context of two-dimensional streamline plots 
are further illustrated 
by correspondig simulations of the three-dimensional model. 
Figure \ref{fig:3D_sxzu} presents 
approximate realizations of the fluctuation field 
in Model~\ref{model:turb_inhom} 
at a fixed time point $t$ 
in the form of heat maps and vector plots 
associated to cross sections in a 
three-dimensional rectangular domain.  
The employed 
settings 
are identical to those used in the streamline plots in 
Figures \ref{fig:2D_sx}--\ref{fig:2D_su}, 
with the exceptions that 
the scaling functions and underlying flow fields now depend on the spatial argument 
$\x=(x_1,x_2,x_3)$ in the three-dimensional domain $[0,1]\times[0,0.5]\times [0,0.5]$ 
and that the numerical parameter related to the number of quadrature points is chosen as 
$N = 20\,000$. 
The 
heat maps show the values of 
the first a component $u'_1(\x,t)$ 
of the fluctuation field $\u'(\x,t)$, 
and the length of the vectors depicted in the vector plots is  scaled by a factor of  
$0.01$
in order to facilitate their display. 
Comparing the three-dimensional plots with the 
respective 
two-dimensional streamline plots,
it can be noted that the main features observed in the simplified two-dimensional setting carry over to the three-dimensional case. 
The simulation presented at the top of Figure~\ref{fig:3D_sxzu} 
is based on the scenario from Figure~\ref{fig:2D_sx}, i.e., 
the flow fields are chosen such that 
the spatial scaling function 
$\sx$ 
increases along the $x_1$-axis 
by a factor of four, 
while $\su$ and $\sz$ are constant with value one.    
The 
corresponding 
increase in size of the turbulent structures from left to right 
is visible in both the heat map and the vector plots. 
It is instructive to observe that 
if one partitions 
the quadratic 
cross section on the left 
into 
smaller 
squares 
of one fourth the 
side length of the cross section, 
then 
the composition of the turbulent structures in each of the smaller squares 
appears to be 
similar to 
the composition of the structures on 
the whole quadratic cross section on the right, 
reflecting 
the growth of $\sx$ and the constancy of $\su$, $\sz$. 
In the simulation in the middle of Figure~\ref{fig:3D_sxzu} 
the scaling function $\sz$ associated to the inverse turbulent viscosity ratio 
grows 
from left to right 
by a factor of ten, 
whereas all other scaling functions remain constant, 
just as in the scenario from  Figure~\ref{fig:2D_sz}. 
Consistent with 
this setting,  
the size of the large-scale structures 
on 
the quadratic cross-section on the left 
is comparable to the size of the large-scale structures on the right,
while the structures on the left are 
superposed by a wider variety of small-scale structures
than those on the right. 
The simulation presented at the bottom of  
Figure~\ref{fig:3D_sxzu} in turn  adopts the scenario from Figure~\ref{fig:2D_su}, 
in which the velocity scaling factor $\su$ increases along the $x_1$-axis by a factor of four 
and $\sx$, $\sz$ are kept constant.  
Unlike in 
the corresponding streamline plot, 
the growth of the velocity values in magnitude 
is apparent in both the heat map and the vector plots.

\begin{figure}[h!]
	\includegraphics[width=.8\textwidth]{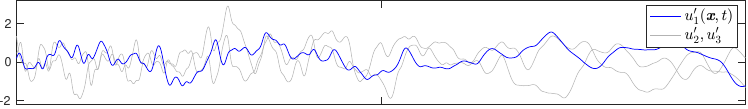}\vfill\vspace{0.3cm}
	\includegraphics[width=.8\textwidth]{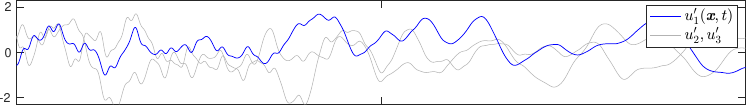}\vfill\vspace{0.3cm}
	\includegraphics[width=.8\textwidth]{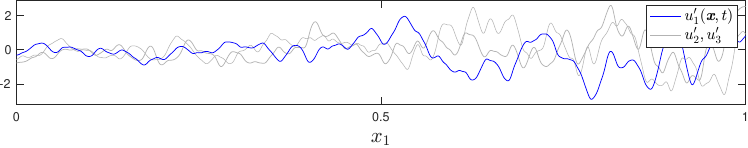}\vfill
	\caption{
	Approximate sample paths of the components 
	$u'_i(\x,t)$ 
	of the turbulence field in Model~\ref{model:turb_inhom} 
	on the line segment $\{\x=(x_1,0,0)\colon 0\leq x_1\leq 1\}$, 
	at a fixed time point $t$. 
	Top: 
	Increasing spatial scaling factor $\sx=k^{3/2}/\eps$ (scenario from Figure~\ref{fig:2D_sx}).
	Middle: 
	Increasing 
	viscosity 
	scaling factor $\sz=\eps\nu/k^2$ (scenario from Figure~\ref{fig:2D_sz}). 
	Bottom: 
	Increasing velocity scaling factor $\su=k^{1/2}$ (scenario from Figure~\ref{fig:2D_su}). 
	See the text for details. 
	}
	\label{fig:paths_sxzu}
\end{figure}

The illustration of the behavior of the three-dimensional model in the considered scenarios 
is complemented in Figure~\ref{fig:paths_sxzu}, 
where approximate sample paths of the 
velocity components 
$u'_1$, $u'_2$, $u'_3$ 
along the $x_1$-axis 
are presented for each scenario.  
The 
underlying settings for the model parameters 
are the same as before, 
except that the value of the numerical parameter 
related to the number of quadrature points is 
taken as 
$N= 4\,000$. 
The discussed 
distinctive 
features of the different scenarios can be observed here too --
be it the growing spatial scale of the fluctuations in the case of increasing $\sx$, 
the changing structural composition of the fluctuations 
in the case of varying $\sz$,  
or the direct influence of $\su$ on 
the magnitude of the fluctuations.

\begin{figure}[h!]
	\centerline{\includegraphics[width=.8\textwidth]{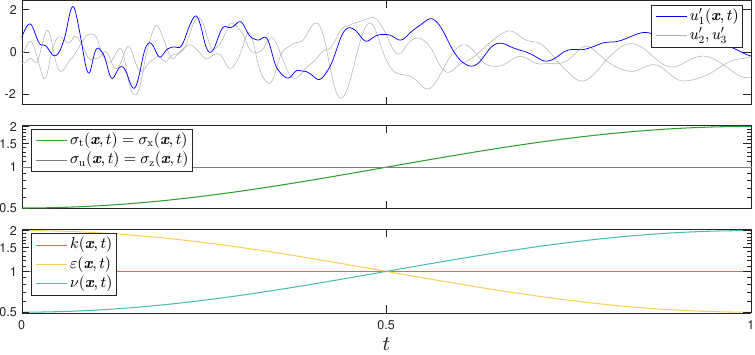}}
	\caption{
	Approximate sample paths 
	\wrt time 
	of the components 
	$u'_i(\x,t)$ 
	of the turbulence field in Model~\ref{model:turb_inhom} 
	at a fixed spatial point $\x$,  
	with temporal scaling 
	function 
	$\st=k/\eps$ 
	increasing along the $x_1$-axis  
	by a factor of four. 
	The scaling functions and underlying flow quantities are depicted in semi-log plots.  
	See the text for details. 
	}
	\label{fig:paths_st}
\end{figure}

As all simulation results discussed up to this point 
concern snapshots of the fluctuation field 
at a fixed point in time, 
the temporal  scaling factor $\st=k/\eps$ has not played a relevant role yet.  
The simulation presented in Figure~\ref{fig:paths_st}  
therefore 
focuses on the temporal evolution of the fluctuations, 
showing approximate sample paths of the velocity components 
as in Figure~\ref{fig:paths_sxzu} but \wrt time instead of space.  
Here we suppose that the temporal scaling factor varies over time 
and grows from $t=0$ to $t=1$ by a factor of four, 
whereas $\su$ and $\sz$ are constant with value one. 
The underlying flow fields are chosen as 
$k(\x,t)=1$, $\eps(\x,t)=2^{\;\!-\!\cos(\pi+\pi t)}$, and $\nu(\x,t)=2^{\;\!\cos(\pi+\pi t)}$, 
so that the temporal scaling factor 
is given by 
$\st(\x,t)= 2^{\;\!\cos(\pi+\pi t)}$ 
and coincides with the spatial scaling factor $\sx$. 
The mean velocity field 
 $\overline \u$ 
is assumed to be identically zero, 
and we further set $\delta=0.08$, $z=0.005$, $N=4\,000$ as before. 
It is clearly noticeable 
that the increase of $\st$ induces a respective decrease of the frequency of 
fluctuations in the sample path plot.  
In combination with the previous simulation results, 
this illustrates how the inhomogeneous random field model 
manages to incorporate the 
flow data provided by the underlying 
fields $k$, $\eps$, and $\nu$ 
into the spatio-temporal 
structure of the generated fluctuations.

\begin{figure}[h!]
	\includegraphics[width=.7\textwidth]{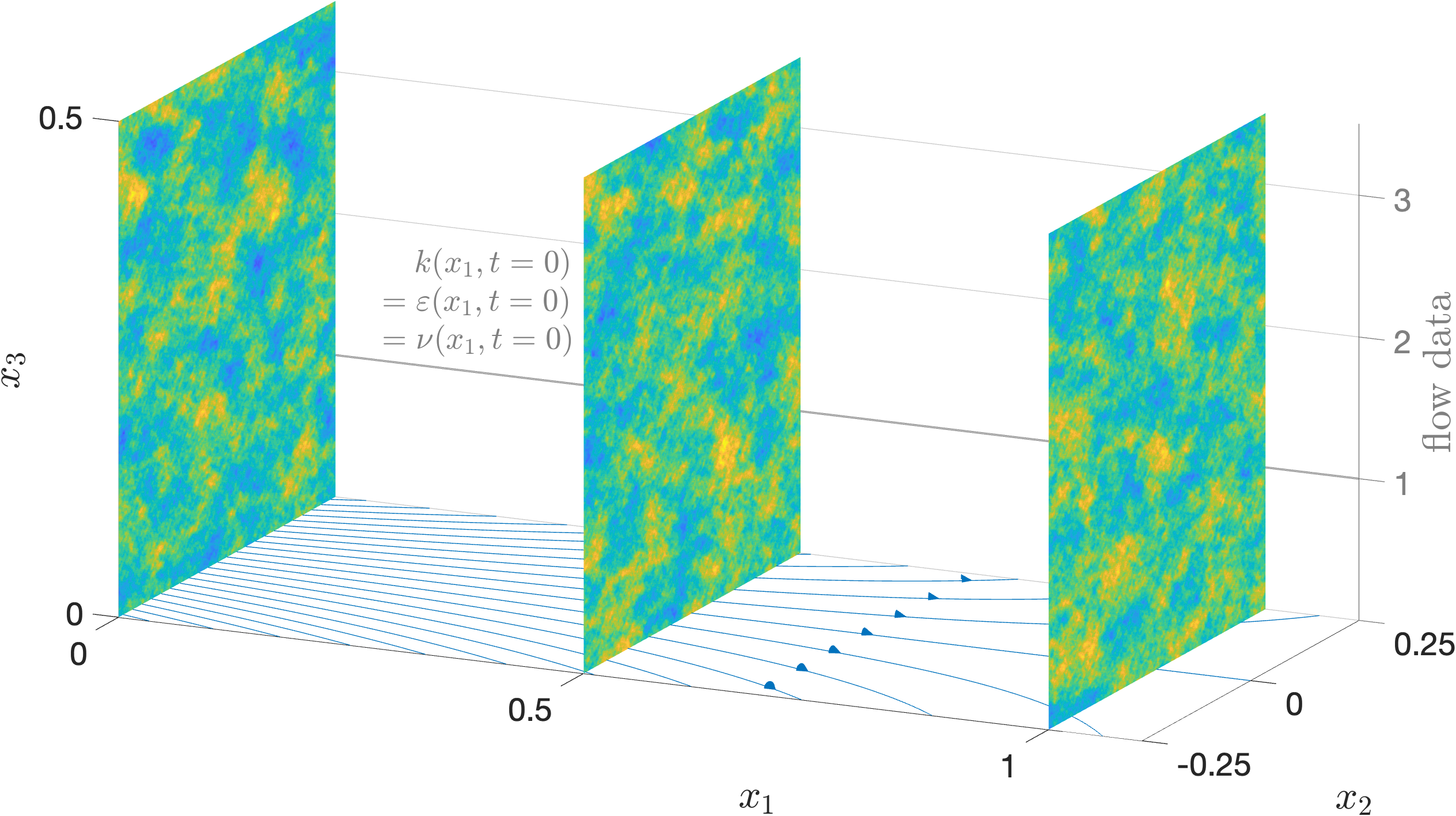}\vfill
	\includegraphics[width=.7\textwidth]{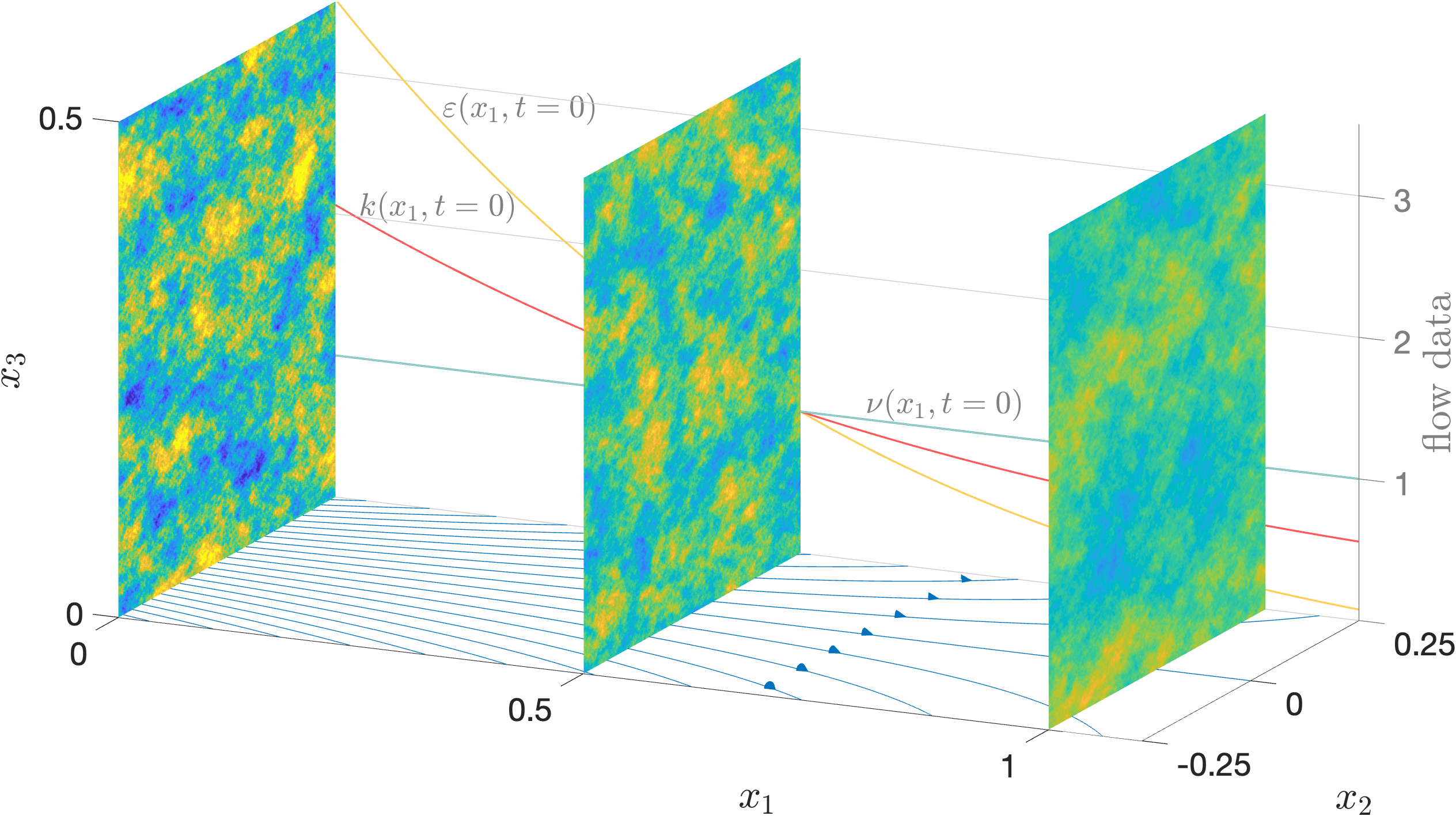}\vfill
	\caption{
	Approximate realizations of the first component $u'_1$ of the turbulence field 
	in Model~\ref{model:turb_inhom} at a fixed time point. 
	The fluctuations are subject to advection by a stationary non-uniform mean flow (streamlines in blue color) 
	in two different szenarios for the underlying flow data $k$, $\eps$, and $\nu$. 
	See the text for details.  
	}
	\label{fig:advection}
\end{figure}

A further 
essential 
feature of our model 
is 
its ability to 
consistently capture 
the advection of the turbulent structures by the mean flow 
even in the case of a non-uniform mean velocity.
Note that the mean velocity field $\overline \u$ enters 
the definition of the fluctuation field $\u'$ in Model~\ref{model:turb_inhom} 
via the mean flow function $\flow$ 
in \eqref{eq:mean_flow}  
for the purpose of describing 
the advection 
from an Eulerian perspective; 
compare the discussion in \cite[Section~5.1]{AKLMW24}. 
In Discretization~\ref{discretization} the 
mean flow function 
appears in a simplified, locally linearized form. 
Figure~\ref{fig:advection} 
highlights the 
feature of non-uniform advection 
by presenting 
two simulations of the fluctuation field $\u'$ 
based on a stationary underlying mean velocity field $\overline\u$ 
of the form 
$\overline u_1(\x,t)=(x_1-1.5)^2$, $\overline u_2(\x,t)=-2(x_1-1.5)x_2$,  $\overline u_3(\x,t)= 0$,    
$\x=(x_1,x_2,x_3)\in [0,1]\times[-0.25,0.25]\times [0,0.5]$. 
Associated mean flow streamlines are 
indicated 
in blue color on the 
$x_1$$x_2$-plane. 
Each of the two plots shows heat maps of an approximate realization 
of the first velocity component $u'_1(\x,t)$ 
on three cross sections parallel to the 
$x_2$$x_3$-plane 
at a fixed time point $t$.  
As indicated 
on the background plane, 
the upper plot represents a stylized szenario 
in which the flow data $k$, $\eps$, $\nu$ are held constant with value one, 
while the szenario of the lower plot assumes $k$ and $\eps$ to be decreasing functions of the $x_1$-coordinate. 
In both simulations the remaining parameters are taken as $\delta=0.08$, $z=10^{-4}$, $N=4\,000$. 
Focusing first on the upper plot, 
we remark that the mean flow essentially transports the turbulent structures 
from left to right, 
so that the structures visible on the cross sections in the middle and on the right 
can be 
roughly 
thought of as having evolved from structures previously located at the cross section on the left. 
As the mean flow also exhibits a nonlinear diverging behaviour, 
the turbulent structures would be streched in $x_2$-direction while being transported along the mean flow stream lines  
if no temporal decay was involved (frozen turbulence).  
This streching effect is prevented due to the natural temporal decay of the turbulent structures induced by 
the time integration kernel $\eta$ in Model~\ref{model:turb_inhom}. 
The generated fluctuations thus remain consistent with the prescribed flow data 
independently of the 
nonlinear transport 
by the mean flow. 
This aspect is further illustrated in the lower plot, where the inhomogeneity of $k$ and $\eps$ 
leads to
corresponding differences
of the turbulent structures on the cross sections \wrt 
velocity magnitude, spatial scale, and spectral composition.

\subsection{Spatio-temporal ergodicity} 
\label{sec:ergodicity}

The 
spatio-temporal 
ergodicity properties of
the inhomogeneous turbulence field in Model~\ref{model:turb_inhom} 
have been analytically investigated in \cite[Theorem~5.6]{AKLMW24}. 
Roughly speaking, the derived results state that local characteristic values of the fluctuation field 
at a point $(\x,t)$ that are given in terms of expected values \wrt the probability distribution of $\u'(\x,t)$ 
can be estimated by means of 
local averages in space and time.  
In this subsection we confirm and illustrate 
the ergodicity properties of our model 
via numerical simulations. 
We focus on the characteristic values of 
turbulent kinetic energy $k(\x,t)$ and dissipation rate $\eps(\x,t)$, 
whose representations in terms of expected values read 
\begin{align}\label{eq:k_eps_expectation}
 \frac{1}{2}\,\E \Bigl[ \bigl\| \u'(\x,t) \bigr\|^2 \Bigr] = k(\x,t), \qquad
\frac12 \,\delta^2 z  \, \E \Bigl[ \bigl\| \nabla_{\x} \u'(\x,t) + \bigl( \nabla_{\x}\u'(\x,t) \bigr){\vphantom{)}}^{\!\top} \bigr\|^2 \Bigr]  \approx \frac{\eps(\x,t)}{\nu(\x,t)};
\end{align}
compare \cite[Theorem~5.5]{AKLMW24} and  Corollary~\ref{prop:flow_properties} above. 
The approximate identity concerning $\eps(\x,t)$ is unterstood as an asymptotic result for 
the turbulence scale ratio $\delta\to0$.
In the following 
simulations we estimate $k(\x,t)$ and $\eps(\x,t)$ 
by replacing the expected values in \eqref{eq:k_eps_expectation} with 
local averages  
of a single sample path in space and time.  
Here the turbulence scale ratio, the inverse turbulent viscosity ratio, 
and the numerical parameter associated to the number of quadrature points 
are set to $\delta=0.01$, $z=0.005$, and $N= 4\,000$.
As before, we make the simplifying assumption of isotropic one-point velocity correlations and take the anisotropy factor $\bm{L}(\x,t)$ 
to be the identity matrix.

\begin{figure}[h!]
	\includegraphics[width=.49\textwidth]{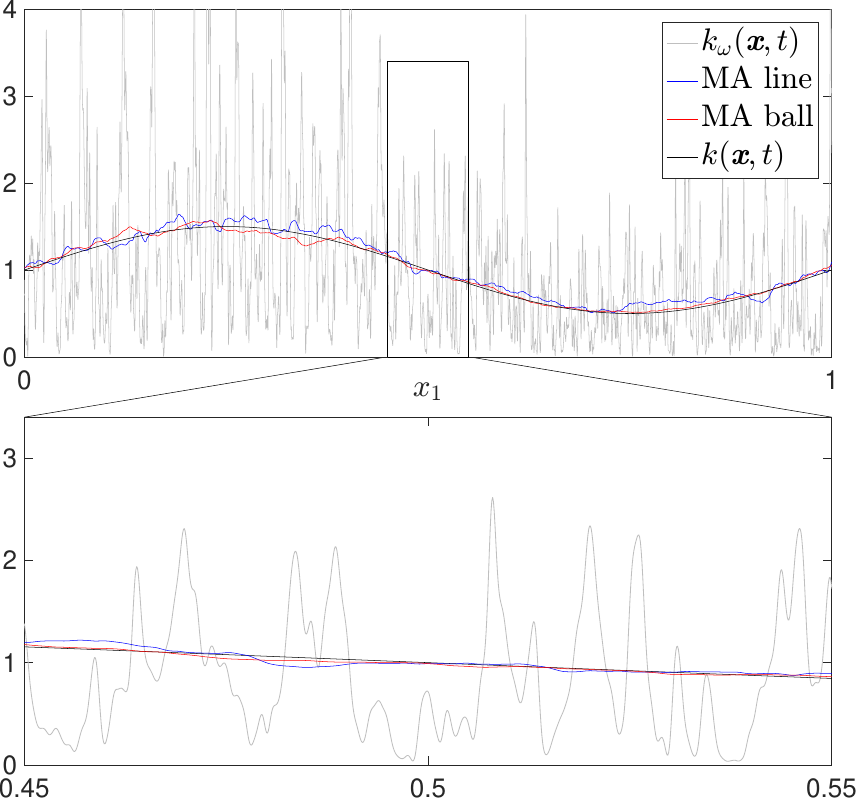}\hfill
	\includegraphics[width=.49\textwidth]{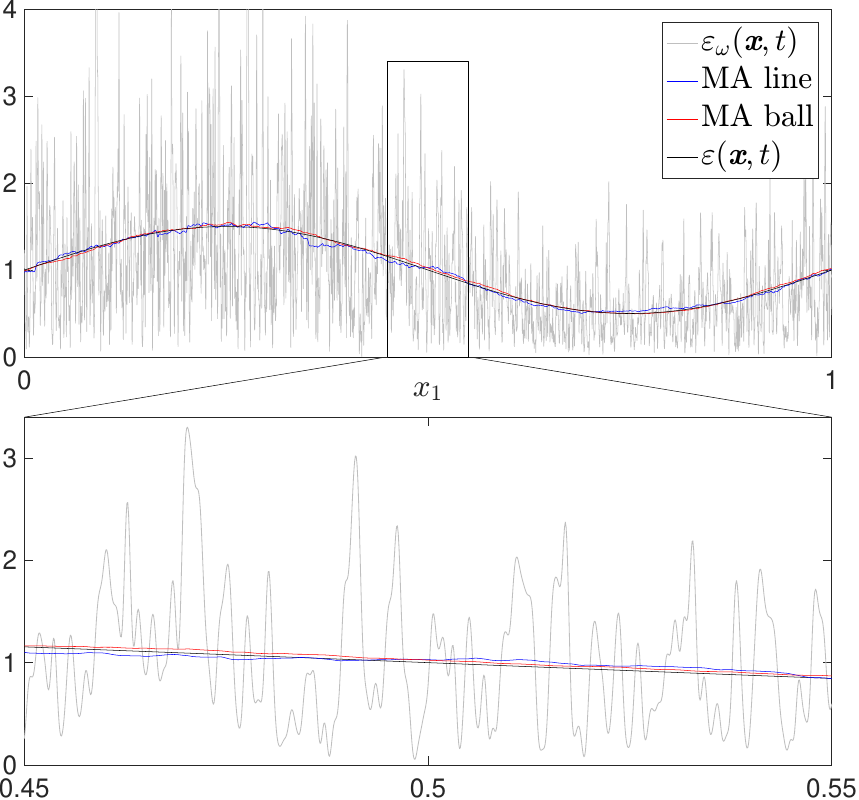}\hfill
	\caption{
			Spatial averaging at a fixed time point $t$. 
			Left: 
			Approximate sample path $k_\omega$ of 
			the instantaneous turbulent kinetic energy 
			$\|\u'\|^2/2$ along the $x_1$-axis,  
			together with associated 
			moving averages 
			over line segments and three-dimensional balls,  
			estimating 
			the underlying turbulent kinetic energy $k$. 
			Right: 
			Approximate sample path $\eps_\omega$ of 
			the instantaneous dissipation rate 
			$\delta^2 z  \,\nu\, \| \nabla_{\x} \u' + ( \nabla_{\x}\u' ){\vphantom{x}}^{\!\top} \|^2/2$, 
			together with associated 
			moving averages 
			estimating 
			the underlying dissipation rate $\eps$. 
			See the text for details. 
			}
	\label{fig:ergo_space}
\end{figure}

We first focus on 
spatial averages 
and consider simulations of the fluctuation field 
at a fixed time point $t$. 
The plots on the left-hand side and on the right-hand side of Figure~\ref{fig:ergo_space} 
involve spatial variations of 
the turbulent kinetic energy and the dissipation rate, 
respectively, 
and present estimates of 
the varying characteristic flow quantities 
in terms of spatial moving averages  
based on sample paths of the fluctuation field. 
More precisely, 
in the szenario underlying the plot on the left-hand side of Figure~\ref{fig:ergo_space} we assume that 
the turbulent kinetic energy varies along the $x_1$-axis and is given by 
$k(\x,t)=1+\sin(2\pi x_1)/2$ for $\x=(x_1,x_2,x_3)$, 
while the dissipation rate $\eps$ and the kinematic viscosity $\nu$ are constant with value one and the mean velocity  $\overline\u$ is identically zero.  
The plot shows the graph of 
$k(\x,t)$ along the $x_1$-axis 
together with an approximate sample path of 
the instantaneous turbulent kinetic energy $\|\u'(\x,t)\|^2/2$, 
denoted by $k_\omega(\x,t)$ in order to emphasize 
its dependence on a random outcome $\omega$ of the underlying probability space  $(\Omega,\mathscr{F},\Pr)$.
In addition, 
two moving 
averages 
associated to the sample path are shown: 
The first moving average consists of average values over line segments on the $x_1$-axis 
and assigns to every value of $x_1$ the 
unweighted average 
of $k_\omega((y,0,0),t)$ 
over all evaluation points $y$ in the interval $[x_1-R,x_1+R]$,  where $R=0.075$. 
The second moving average employs average values over three-dimensional balls 
and assigns to every value of $x_1$ the 
unweighted average 
of $k_\omega(\y,t)$ 
over all evaluation points $\y$ in the closed ball in $\R^3$ with center $(x_1,0,0)$ and radius 
$R=0.075$. 
The evaluation points are taken from a rectangular grid with 
a 
fine spacing in $x_1$-direction 
($\Delta x_1= 10^{-4}$)  
and 
a 
coarser spacing in  $x_2$-direction and $x_3$-direction ($\Delta x_2= \Delta x_3 = 0.032$). 
We remark that 
the 
second type of average can be interpreted as an approximation of the 
average integrals considered in 
\cite[(5.42)]{AKLMW24},  
with $R'=0.075$ and $R=0$ in the notation used therein. 
It its clearly visible that the sample path $k_\omega$ oscillates around the mean function $k$ and that the latter is approximately reproduced by the 
local average values. 
The moving average 
based on three-dimensional balls shows a slightly better fit than the one based on line segments. 
As a side observation, note that the sample path $k_\omega$ 
further illustrates  
that the magnitude of $k$ influences 
not only the amplitude of the turbulent fluctuations 
but also their 
spatial frequency governed by the scaling factor $\sx=k^{3/2}/\eps$. 
The plot on the right-hand side of Figure~\ref{fig:ergo_space} presents an analogue simulation
with regard to the 
dissipation rate $\eps$. 
Here we assume that 
$\eps(\x,t)=1+\sin(2\pi x_1)/2$, $\x=(x_1,x_2,x_3)$, 
and keep the turbulent kinetic energy $k$ constant with value one. 
As before, the kinematic viscosity $\nu$  and the mean velocity  $\overline\u$ are constant with values one and zero, respectively.   
The plot shows an approximate sample path $\eps_\omega(\x,t)$ of 
the instantaneous dissipation rate 
$\delta^2 z  \,\nu\, \| \nabla_{\x} \u'(\x,t) + ( \nabla_{\x}\u'(\x,t) ){\vphantom{x}}^{\!\top} \|^2/2$ 
along the $x_1$-axis,  
together with associated moving averages 
over line segments and 
balls that are 
calculated in the same way as before. 
Here the averages over three-dimensional balls can be considered as approximations of the 
average integrals in
\cite[(5.43)]{AKLMW24},  
with $R'=0.075$ and $R=0$ in the notation used therein. 
Compared to the plot on the left-hand side, 
the sample path $\eps_\omega$ appears to oscillate more rapidly than $k_\omega$, 
which reflects the fact that the instantaneous dissipation rate 
is composed of a larger variety of fluctuating terms 
stemming from the 
different 
components of the velocity gradient. 
Accordingly, 
the 
moving averages approximate the mean function $\eps$ 
even better than in the case of the turbulent kinetic energy $k$.

\begin{figure}[h!]
	\includegraphics[width=.49\textwidth]{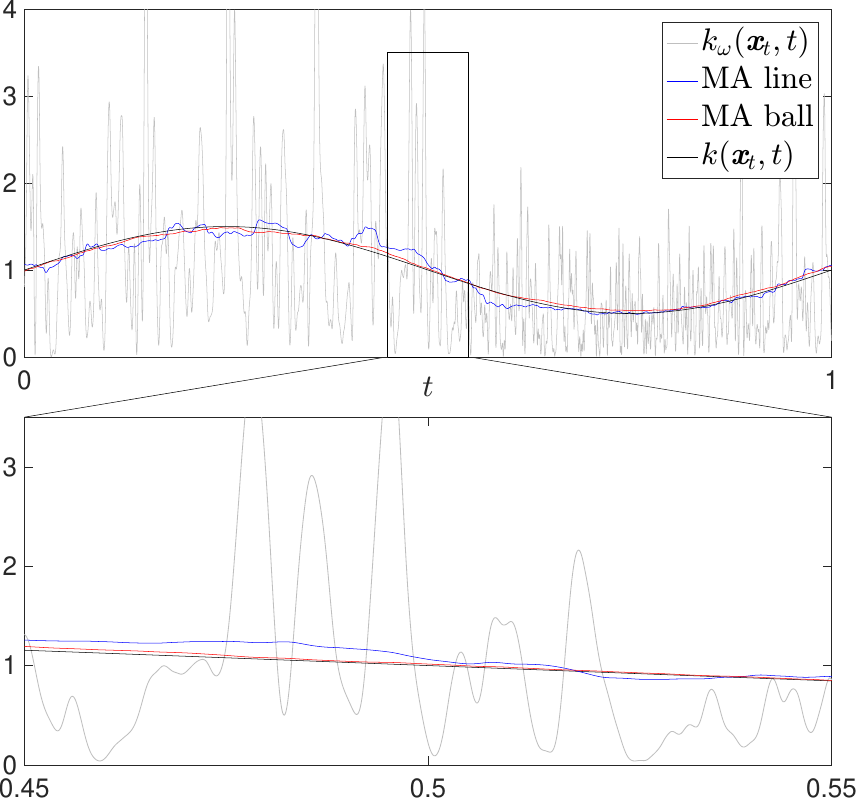}\hfill
	\includegraphics[width=.49\textwidth]{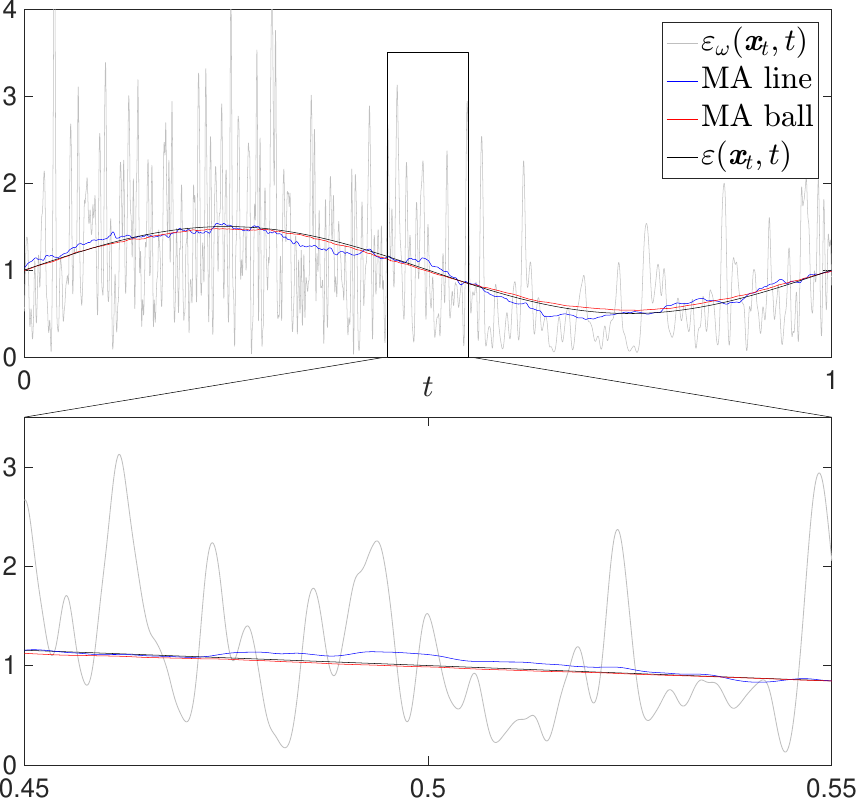}\hfill
	\caption{
			Spatio-temporal averaging. 
			Left: 
			Approximate sample path $k_\omega$ of 
			the instantaneous turbulent kinetic energy 
			$\|\u'\|^2/2$ along a mean flow  
			pathline $\x_t$,  
			together with associated 
			moving averages 
			over pathline segments and environments thereof generated by three-dimensional balls, 
			estimating 
			the underlying turbulent kinetic energy $k$. 
			Right: 
			Approximate sample path $\eps_\omega$ of 
			the instantaneous dissipation rate 
			$\delta^2 z  \,\nu\, \| \nabla_{\x} \u' + ( \nabla_{\x}\u' ){\vphantom{x}}^{\!\top} \|^2/2$, 
			together with associated 
			moving averages 
			estimating 
			the underlying dissipation rate $\eps$. 
			See the text for details. 
			}
	\label{fig:ergo_time}
\end{figure}

Next we include temporal averages and investigate simulations of the fluctuation field over time.  
As in the previous scenarios,  
the plots on the left-hand side and on the right-hand side of Figure~\ref{fig:ergo_time} 
involve spatial variations of the turbulent kinetic energy and the dissipation rate, respectively, 
but now 
the fluctuations are captured from the perspective of an observer moving with the mean flow as time evolves. 
This perspective allows to focus on those temporal fluctuations in our model 
that are due to the temporal decay of the turbulent structures induced by the time integration kernel $\eta$, 
without  
superpositions due to the advection of spatial fluctuations by the mean flow. 
We assume a uniform mean velocity  $\overline\u(\x,t)=(1,0,0)$ 
and employ the mean flow pathline $\x_t=(t,0,0)$. 
In the szenario underlying the plot on the left-hand side of Figure~\ref{fig:ergo_time} the turbulent kinetic energy varies according to 
$k(\x,t)=1+\sin(2\pi x_1)/2$, $\x=(x_1,x_2,x_3)$, 
while the dissipation rate $\eps$ and the kinematic viscosity $\nu$ are constant with value one. 
The plot shows an approximate 
sample path $k_\omega(\x_t,t)$ 
of the instantaneous turbulent kinetic energy $\|\u'(\x_t,t)\|^2/2$  
along the mean flow pathline 
$\x_t$ 
together with two associated moving averages.  
The first moving average consists of average values over pathline segments 
and assigns to every 
time point $t$ 
the unweighted average of $k_\omega(\x_s,s)$ 
over all evaluation time points $s$ in the interval $[t-R,t+R]$. 
The second moving average is taken over 
environments of pathline segments 
based on three-dimensional 
balls and 
assigns to every time point $t$ the 
unweighted average 
of $k_\omega(\y,s)$ 
over all spatio-temporal evaluation points $(\y,s)$ satisfying 
$\max\{\|\y-\x_s\|,|s-t|\}\leq R$. 
The size of the averaging domains is specified by $R=0.075$,  
and the evaluation points are taken from a spatio-temporal grid with 
a fine resolution on the time axis ($\Delta t = 10^{-4}$) 
and a coarser resolution on the spatial axes ($\Delta x_i= 0.032$ for $i=1,2,3$). 
Both types of averages 
can be interpreted as approximations of the 
average integrals in \cite[(5.42)]{AKLMW24},  
choosing $R'=0$ in the notation used therein 
for the first type of average 
and $R'=R$ for the second type. 
Similar to the purely spatial simulation in Figure~\ref{fig:ergo_space}, it can be 
observed 
that the temporal sample path $k_\omega$ oscillates 
around the mean function $k$, 
allowing the temporal and spatio-temoral averages to approximately recover the latter. 
The significantly better fit of the second moving average is owed to the 
fact that it employs information of both temporal and spatial fluctuations. 
As a side note, we remark that 
the visible change in frequency of the oscillations 
in $k_\omega$ 
over time 
is in full accordance with the varying magnitude of $k$ 
and its influence on the turbulent time scale governed by the scaling factor $\st=k/\eps$. 
The fact that the temporal sample path $k_\omega$ appears to have a slightly simpler structure than 
its spatial counterpart in Figure~\ref{fig:ergo_space}  
is due to the  simplified treatment of temporal correlations in 
our model as opposed to the more detailed description of the spatial spectral structure. 
 The plot on the right-hand side of Figure~\ref{fig:ergo_time} presents analogue simulation results 
with regard to the dissipation rate $\eps$,  
assuming that 
$\eps(\x,t)=1+\sin(2\pi x_1)/2$ for $\x=(x_1,x_2,x_3)$,  
while $k$ and $\nu$ are constant with value one.    
It shows an approximate sample path $\eps_\omega(\x_t,t)$ of 
the instantaneous dissipation rate 
$\delta^2 z  \,\nu\, \| \nabla_{\x} \u'(\x_t,t) + ( \nabla_{\x}\u'(\x_t,t) ){\vphantom{x}}^{\!\top} \|^2/2$ 
along the mean flow pathline $\x_t$,  
together with associated moving averages 
over pathline segments and 
spatio-temporal environments thereof that are
calculated in the same way as before. 
Both types of averages  
represent approximations of the average integrals considered in \cite[(5.43)]{AKLMW24},  
using $R'=0$  and $R'=R$ in the notation therein. 


\subsection{Kolmogorov's two-thirds law} 
\label{sec:Kolmogorov}

In view of the fact that 
our random field model 
is focused 
on second-order statistics 
provided 
by case-dependent RANS data 
as well as general turbulence theory, 
we complement the previous discussions 
by examining to what extent 
Kolmogorov's two-thirds law 
for the second moments of the spatial velocity increments 
can be reproduced in an inhomogeneous context.  
Employing the model spectrum from \cite[Example~2.3]{AKLMW24} as before, we 
demonstrate that the 
inhomogeneous 
synthetic turbulence field 
captures 
the two-thirds scaling law 
in dependence on the 
local value of the 
turbulent viscosity ratio (turbulence Reynolds number)
and even realizes 
the correct constant of proportionality  
in consistency with analytical estimates. 

In order to reformulate 
the two-thirds law 
within the framework of our two-scale approach, 
we first consider the idealized case of 
homogeneous and isotropic turbulence 
with 
Kolmogorov spectrum  
$E(\kappa)=C_\mathrm{K}\,\kappa^{-5/3}$ 
on the whole range of wave numbers $\kappa$. 
Note that in this case only the velocity increments $\u'(\x+\bm{r},t)- \u'(\x,t)$ 
are defined but not the turbulent velocity $\u'(\x,t)$  itself. 
For the sake of consistency 
with our 
non-dimensionalization and scaling 
approach, compare Section~\ref{sec:model},   
the  power-law spectrum does not carry the usual scaling factor 
$\eps^{2/3}$.  
The Kolmogorov constant $C_\mathrm{K}$ is taken as $C_\mathrm{K}=1.5$. 
Under these idealized assumptions, 
the two-thirds 
law can be directly derived 
from the representation formula~\eqref{eq:u'}  
in combination with the integral isometry~\eqref{eq:isometric_property} 
and takes the form 
\begin{equation}\label{eq:long}
\E\Bigl[\Bigl( \bigl[\u'(\x+r\bm{e},t)- \u'(\x,t)\bigr]\cdot
\bm{e}
\Bigr)^{\!2}\;\!\Bigr]
\,=\,
C_l \,\Bigl(\frac{\eps r}\delta\Bigr)^{2/3}, 
\end{equation}  
$r\in \R^+$, with constant of proportionality given by 
\begin{equation}\label{eq:constant_long}
C_l
= 
C_{\mathrm{K}} \,\frac{36}{55} \int_0^\infty \!\kappa^{-5/3}(1-\cos(\kappa))\,\d\kappa \approx 
1.97.
\end{equation}
The left-hand side of \eqref{eq:long},  where 
$\bm{e}\in\R^3$ 
is an arbitrary unit vector, 
represents the second-order longitudinal structure function. 
The appearance of the turbulence scale ratio $\delta$ 
on the right-hand side of \eqref{eq:long} 
is owed to the fact that the spatial argument of $\u'$  
is scaled \wrt the macroscopic reference scale, 
whereas the two-thirds law is naturally connected to the turbulence scale. 
In a similar way one obtains the relation
\begin{equation}\label{eq:trans}
\E\Bigl[\Bigl( \bigl[\u'(\x+r\bm{e},t)- \u'(\x,t)\bigr]\cdot
\bm{e}^\perp 
\Bigr)^{\!2}\;\!\Bigr]
\,=\,
\frac 43\; \E\Bigl[\Bigl( \bigl[\u'(\x+r\bm{e},t)- \u'(\x,t)\bigr]\cdot
\bm{e}
\Bigr)^{\!2}\;\!\Bigr] 
\end{equation} 
between the transversal structure function and the longitudinal structure function, 
where $\bm{e}^\perp$ is a unit vector orthogonal to $\bm{e}$.
The identities 
\eqref{eq:long}, \eqref{eq:constant_long}, and  \eqref{eq:trans}  
are well-known 
in the literature  
and correspond to, e.g., \cite[(6.30), (6.245), and (6.31)]{Pope00}. 

Turning back to the 
case of inhomogeneous flow conditions 
and to the model spectrum from \cite[Example~2.3]{AKLMW24}, 
in which the 
$-5/3$ power law 
is restricted 
to a limited range of wavenumbers (inertial subrange), 
we examine whether the idealized identities \eqref{eq:long}--\eqref{eq:trans} 
can be recovered in an approximate sense on a suitable range of distances $r$. 
To this end, we consider Monte Carlo estimates 
for the rescaled longitudinal and transversal structure functions 
\begin{equation}\label{eq:longtrans}
\Bigl(\frac{\eps_r\;\!r}\delta\Bigr)^{\!-2/3}\,\E\bigl[\bigl(u'_1(r\bm{e}_1,0) - u'_1(\bm{0},0)\bigr)^{\!2}\;\!\bigr] 
\qquad \text{ and } \qquad 
\Bigl(\frac{\eps_r\;\!r}\delta\Bigr)^{\!-2/3}\,\E\bigl[\bigl(u'_2(r\bm{e}_1,0) - u'_2(\bm{0},0)\bigr)^{\!2}\;\!\bigr],  
\end{equation}
where $\bm{e}_1$ denotes the unit vector in $x_1$-direction 
and $\eps_r := [\eps(r\bm{e}_1,0)+\eps(\bm{0},0)]/2$ 
accounts for the possibility of a non-uniform dissipation rate. 
Since the range of 
distances $r$ 
for which the two-thirds law can be observed 
is expected to depend on the value of the 
inverse turbulent viscosity ratio (inverse turbulence Reynolds number) $\zeta = \sz(\x,t)z$, 
we compare various simulations based on different values of $z$, 
using $z=10^{-9}, 10^{-8},\ldots,10^{-2}$. 
In all presented simulations, 
spatial inhomogeneity enters either through 
the underlying turbulent kinetic energy $k(\x,t)$ 
or through the dissipation rate $\eps(\x,t)$. 
The turbulence scale ratio is set to $\delta=0.08$,  
the kinematic viscosity $\nu$ is 
assumed to be 
constant with value one, 
the mean velocity $\overline\u$ is 
assumed to be 
identically zero, 
and isotropic one-point correlations are imposed 
by taking the anisotropy factor 
$\bm{L}$ to be the identity matrix.  
The numerical parameter governing the number of quadrature points 
is chosen as $N=4\,000$, 
and the expected values in \eqref{eq:longtrans} 
are estimated via ensemble averages  
based on $20\,000$ realizations of the fluctuation field each.

\begin{figure}[h!]
	\includegraphics[width=.45\textwidth]{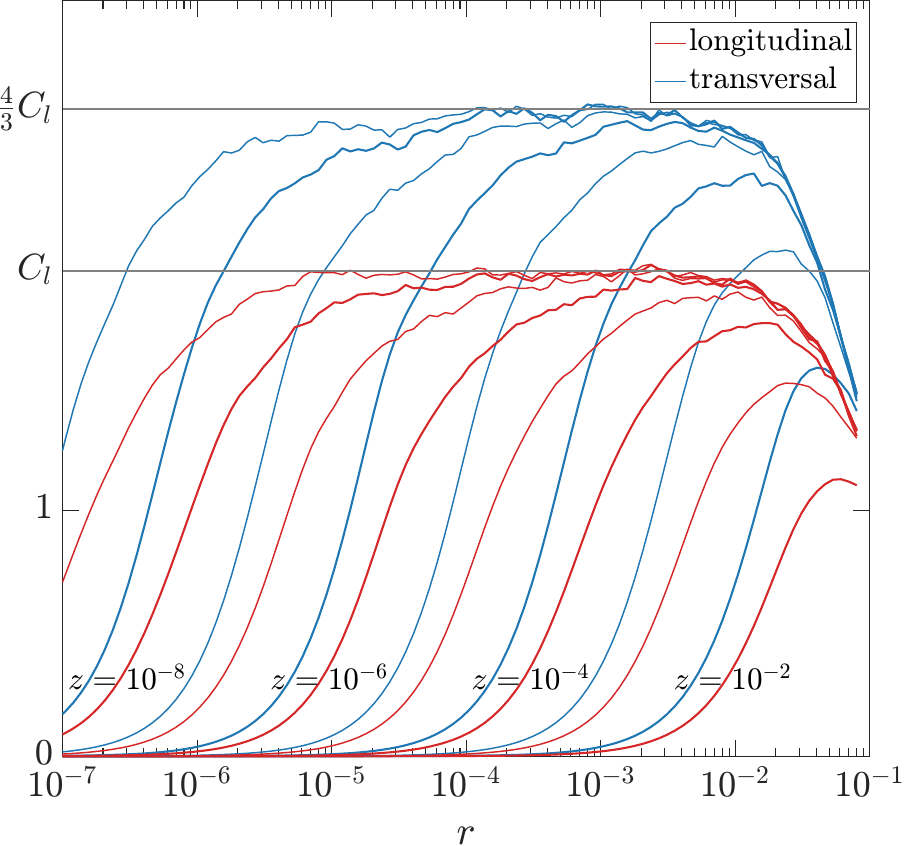} \qquad
	\includegraphics[width=.45\textwidth]{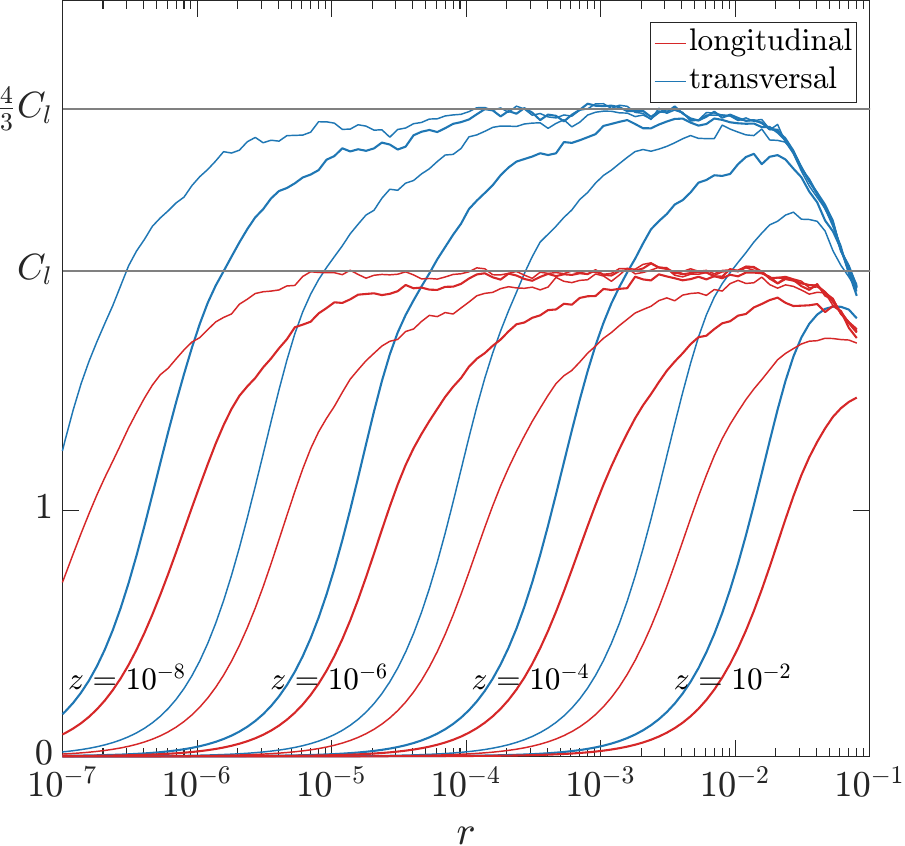} 
	\caption{
	Monte Carlo estimates of the rescaled 
	longitudinal and transversal structure functions~\eqref{eq:longtrans}  
	based on different characteristic 
	values $z$ for the inverse turbulence Reynolds number. 
	The underlying 
	turbulent kinetic energy $k$ increases linearly in $r$ 
	with slope $a=1$ (left) and slope $a=10$ (right). 
	See the text for details. 
	}
	 \label{fig:Kolmogorov1}
\end{figure}

In the scenarios underlying the simulation results 
shown in Figure~\ref{fig:Kolmogorov1},   
the turbulent kinetic energy 
is chosen to 
grow linearly in $x_1\in[0,10^{-1}]$ with slope $a$ 
such that $k(\x,t)=1$ for $x_1=10^{-4}$, 
i.e., $k(\x,t)=a(x_1-10^{-4})+1$.   
The dissipation rate $\eps(\x,t)$ is taken to be constant with value one. 
The plots on the left-hand side and 
on the right-hand side of Figure~\ref{fig:Kolmogorov1} 
correspond to 
the slopes $a=1$ and $a=10$, respectively, 
and depict the 
estimates 
of the rescaled structure functions~\eqref{eq:longtrans}  
for the considered values of $z$. 
The roughly bell-shaped graphs are 
widened towards 
the left as $z$ decreases. 
Note that Kolmogorov's two-thirds scaling 
can indeed be 
approximately 
observed 
in accordance with \eqref{eq:long}--\eqref{eq:trans},  
given that the 
inverse turbulence Reynolds number $\zeta = \sz(\x,t)z$ 
is small enough. 
For larger values of $z$, i.e., for low  Reynolds numbers, 
the inertial subrange is not wide enough for the plotted graphs 
to reach a plateau near the target values $C_l$ and $(4/3)C_l$. 
We remark that, unlike commonly found in the literature, 
the employed model spectrum does not rely on approximation arguments 
involving infinite Reynolds number limits to ensure the fulfillment of the 
essential integral conditions \eqref{eq:energy_spectrum}, 
 so that the low Reynolds number case is consistently addressed as well.   
It is further worth emphasizing that the observed two-thirds scaling occurs in inhomogenous scenarios. 
In fact, 
the stronger inhomogeneity ($a=10$) even leads to 
slightly wider plateaus in the plot on the right-hand side,  
which is consistent with the 
increase and decrease 
of the spatial scaling factors 
$\sx(\x,t)=k^{3/2}(\x,t)/\eps(\x,t)$ and $\sz(\x,t)=\eps(\x,t)\nu(\x,t)/k^2(\x,t)$, 
respectively,  along the $x_1$-axis.

\begin{figure}[h!]
	\includegraphics[width=.45\textwidth]{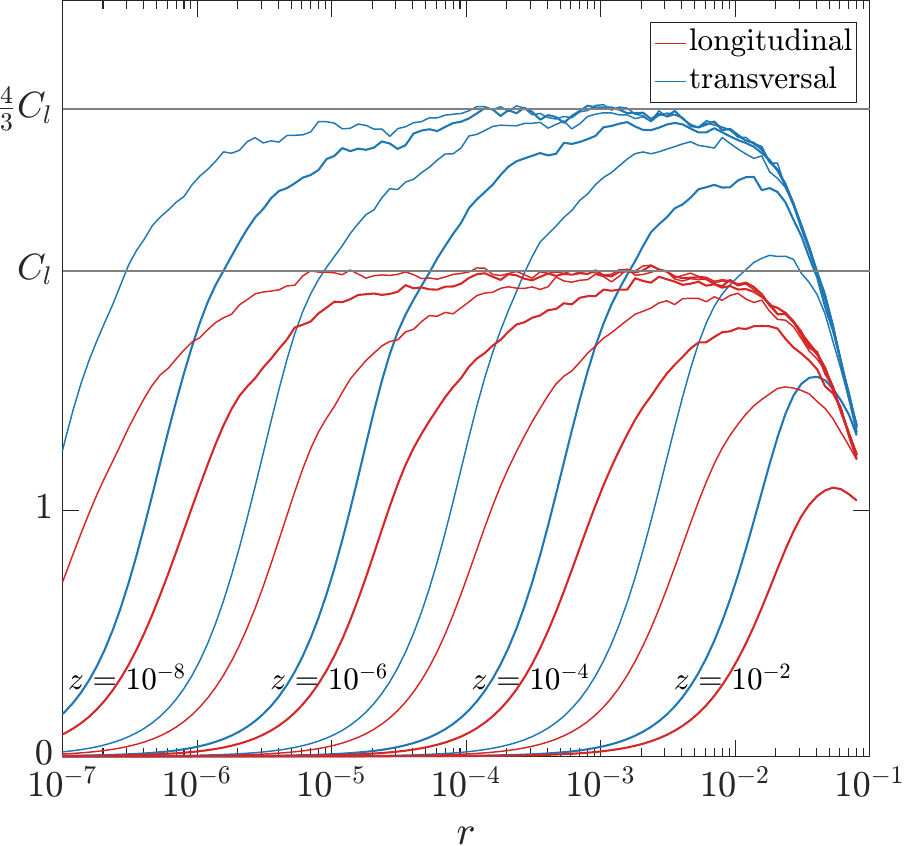} \qquad
	\includegraphics[width=.45\textwidth]{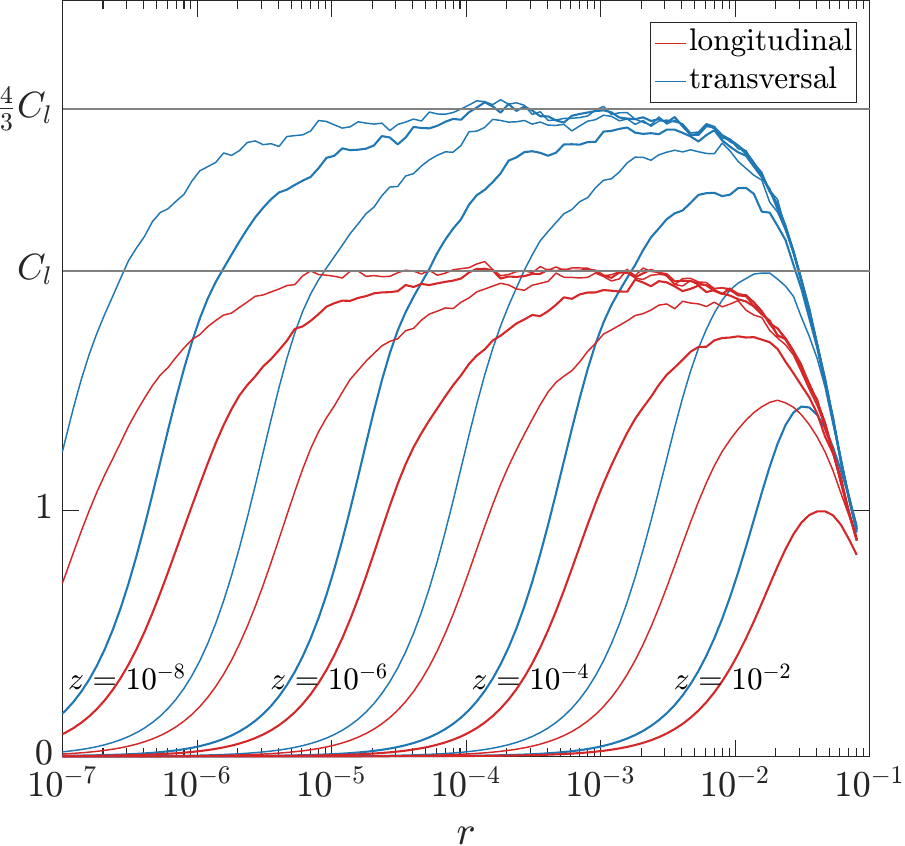} 
	\caption{
	Monte Carlo estimates of the rescaled 
	longitudinal and transversal structure functions~\eqref{eq:longtrans}  
	based on different characteristic 
	values $z$ for the inverse turbulence Reynolds number. 
	The underlying 
	dissipation rate $\eps$ increases linearly in $r$ 
	with slope $a=1$ (left) and slope $a=10$ (right). 
	See the text for details. 
 	}
	 \label{fig:Komogorov2}
\end{figure}

The simulation results presented in Figure~\ref{fig:Komogorov2} 
are based on interchanged assumptions concerning the involved inhomogeneity: 
Here 
the dissipation rate is assumed to grow linearly in $x_1\in[0,10^{-1}]$, 
$\eps(\x,t)=a(x_1-10^{-4})+1$,
while the turbulent kinetic energy $k(\x,t)$ is constant with value one. 
Slopes $a=1$ and $a=10$ are employed for the plots 
on the left-hand side and 
on the right-hand side of Figure~\ref{fig:Komogorov2}, respectively.  
As before, the shown estimates for the rescaled structure functions \eqref{eq:longtrans}  
exhibit Kolmogorov's two-thirds scaling 
in accordance with \eqref{eq:long}--\eqref{eq:trans} 
as long as the inverse turbulence Reynolds number is small enough. 
The plots on the left-hand side of Figure~\ref{fig:Komogorov2} 
are similar to the plots on the left-hand side of Figure~\ref{fig:Kolmogorov1}, 
indicating that the considered 
inhomogeneities with slope $a=1$ 
are not strong enough to 
influence the second-order structure functions in a significant way. 
By contrast, the plots on the right-hand sides of 
Figure~\ref{fig:Kolmogorov1} and Figure~\ref{fig:Komogorov2} 
clearly differ for larger values of $r$. 
This is due to the relatively large variations 
in the employed inhomogeneities with slope $a=10$ 
in combination with the fact that the 
growth 
of $k(\x,t)$ 
underlying 
Figure~\ref{fig:Kolmogorov1} and the 
growth 
of $\eps(\x,t)$  
underlying 
Figure~\ref{fig:Komogorov2} 
lead to directly opposite spatial variations of the 
relevant scaling factors $\sx(\x,t)$ and $\sz(\x,t)$.


\appendix
\renewcommand{\theequation}{\thesection.\arabic{equation}}

\section{Gaussian white noise and stochastic integrals} \label{app:wn}

The turbulent fluctuations $\u'$  are modelled in terms of stochastic integrals \wrt an underlying vector-valued white noise, hence we shortly recall these concepts.

\begin{definition}[Gaussian white noise]\label{def:white_noise}
Let $U \subset \R^m$ be a Borel set, let $\mu \colon \B(U) \to [0,\infty]$ be a $\sigma$-finite measure, and let $\B_0(U)$ denote the system of all sets $A \in \B(U)$ with $\mu(A) < \infty$.  A mapping 
$\wn = (\wnr_{1},\ldots,\wnr_{\ell}) \colon \B_0(U) \to L^2(\Pr;\C^\ell)$ 
is called a \emph{$\C^\ell$-valued Gaussian white noise on $U$ with structural measure $\mu$} 
if the family of real-valued random variables $\Re\;\!\wnr_k(A)$ (real part), $\Im\;\!\wnr_k(A)$ (imaginary part), $A\in \B_0(U)$, $k\in\{1,\ldots,\ell\}$, is jointly Gaussian and for all  $A,B \in \B_0(U)$ we have
\begin{align}\label{eq:white_noise}
\E\bigl[ \wn(A) \bigr] = \bm0, 
\qquad 
\E\bigl[ \wn(A)\otimes\overline{\wn(B)}\bigr] = \mu(A\cap B)\, \bm I, 
\qquad 
\E\bigl[ \wn(A)\otimes\wn(B)\bigr] = \bm 0,
\end{align}
where $\bm I \in \R^{\ell\times\ell}$ denotes the identity matrix.
\end{definition} 
Definition~\ref{def:white_noise} is equivalent to the alternative characterization presented 
in  \cite[Definition~A.1]{AKLMW24}. In particular, the first two conditions in \eqref{eq:white_noise} imply the additivity property 
$\wn(A\cup B)=\wn(A)+\wn(B)$ for disjoint sets $A,B$ 
as well as the fact that the components $\wnr_1,\ldots,\wnr_\ell$ have the same 
structural measure $\mu$ and are uncorrelated in the sense that 
$\E[|\wnr_j(A)|^2]=\mu(A)$ and $\E[\wnr_j(A)\:\!\overline{\wnr_k(A)}]=0$ for $j\neq k$. 
The third condition in \eqref{eq:white_noise} additionally specifies the correlation structure of the real and imaginary parts and ensures that 
$\Re\;\!\wnr_1$, $\Im\;\!\wnr_1,\ldots,\Re\;\!\wnr_\ell$, $\Im\;\!\wnr_\ell$  
 share the  structural measure $\mu/2$ and are uncorrelated as well. 
 As a consequence of these properties,  the stochastic integral $\int_U \bm{G}(\x) \cdot \wn(\d\x)$ is well-defined for integrands $\bm{G}\in L^2(\mu;\C^{d\times\ell})$ as an element of $L^2(\Pr;\C^{d})$ 
satisfying the isometric identities 
 \begin{align}\label{eq:isometric_property}
\E \biggl[ \Bigl\| \, \Re \! \int_U \bm{G}(\x) \cdot \wn(\d\x) \Bigr\|^2 \biggr] 
= 
\E \biggl[ \Bigl\| \, \Im \! \int_U \bm{G}(\x) \cdot \wn(\d\x) \Bigr\|^2 \biggr] 
=
\frac{1}{2}\int_U \bigl\| \bm{G}(\x) \bigr\|^2 \mu(\d\x).
\end{align}
We refer to \cite[Appendix~A]{AKLMW24} and the references therein for further details on stochastic integration \wrt complex vector-valued white noise.

\section{Randomized approximation of  stochastic integrals}\label{app:approx}

The convergence proof for the numerical approximation scheme presented in Section~\ref{sec:discretization_convergence} 
relies on 
general
auxiliary results on 
Monte Carlo quadrature methods for 
stochastic integrals \wrt Gaussian white noise 
established 
below. 
Related results can be found, e.g., in \cite{BK95, PS95, Pri01, LPS14} and the references therein. 
In what follows, the abbreviation ``i.d.d.'' stands as usual for ``independent and identically distributed''. 

We remark that the 
randomized approximation 
method described in 
Proposition~\ref{prop:stratMC} 
involves stratified sampling on the domain of integration via the decomposition $U = \bigcup_j \Delta_j$ as well as importance sampling on each subdomain $\Delta_j$ via the reference density $p_j$.

\begin{proposition}[Monte Carlo integration] \label{prop:stratMC}
Let $U\subset\R^m$ be a Borel set, let $\mu\colon \B(U) \to [0,\infty]$ be a $\sigma$-finite measure, and let $\wn$ be a $\C^\ell$-valued Gaussian white noise on $U$ with structural measure $\mu$. 
Moreover, 
let $\mathcal J$ be a countable index set, 
let $\Delta_j\subset U$, $j\in\mathcal J$, be a measurable partition of $U$ 
(i.e., $\Delta_j\in\mathcal B(U)$, $\bigcup_{j\in\mathcal J}\Delta_j=U$, and $\Delta_j\cap\Delta_k=\emptyset$ for $j\neq k$),   
and let $p_j\colon U\to\R_0^+$, $j\in\mathcal J$, be probability densities \wrt $\mu/2$ 
such that
$\int_{\Delta_j}p_j(\y)\,\mu(\d\y)/2=1$.  
Let 
$\y_{jn}$, $j\in\mathcal J$, $n\in\N$, be 
independent $U$-valued  
random variables with distribution 
\begin{align*}
\y_{jn}\;\!\sim\;\! \frac12\,p_j(\y)\;\!\mu(\d\y)
\end{align*}
and let $\wn_{jn}$, $j\in\mathcal J$, $n\in\N$, be 
i.i.d.\ $\C^\ell$-valued square-integrable random variables such that 
the $\R^\ell$-valued random variables 
$\Re\;\!\wn_{jn}$, $\Im \;\!\wn_{jn}$
are uncorrelated
with mean zero 
and identity covariance matrix 
\begin{align*}
\E\bigl[\Re\;\!\wn_{jn}\bigr]=\E\bigl[\Im\;\!\wn_{jn}\bigr]=\bm0, \quad \E\bigl[\Re\;\!\wn_{jn}\otimes\Re\;\!\wn_{jn}\bigr]=\E\bigl[\Im\;\!\wn_{jn}\otimes\Im\;\!\wn_{jn}\bigr]=\bm I.
\end{align*}
Assume that the 
family $\y_{jn}$, $j\in\mathcal J$, $n\in\N$, is independent 
of 
$\wn_{jn}$, $j\in\mathcal J$, $n\in\N$, 
and for every $N\in\N$ let $\wn_N$ be the $\C^\ell$-valued random measure on $U$ defined by 
\begin{align*}
\wn_N(\;\!\cdot\;\!) = \sum_{j\in\mathcal J}\frac1{\sqrt{N}} \sum_{n=1}^N \frac1{\sqrt{p_j(\y_{jn})}} \, \wn_{jn} \, \delta_{\y_{jn}}(\;\!\cdot\;\!),
\end{align*}
where $\delta_\y$ denotes Dirac measure at $\y \in U$.
Then, for every $\bm{G}\in L^2(\mu;\C^{d\times\ell})$ we have the 
following 
convergence in distribution of $\C^d$-valued random variables 
\begin{align}\label{eq:strat_MC_result}
\int_U \bm{G}(\y)\cdot\wn_N(\d\y) 
= 
\sum_{j\in\mathcal J}\frac1{\sqrt{N}} \sum_{n=1}^N \frac1{\sqrt{p_j(\y_{jn})}} \, \bm{G}(\y_{jn})\cdot\wn_{jn} 
\;\!\xrightarrow{\;\mathrm{d}\;}\;\!
\int_{\bigcup_{j}\{p_j>0\}} \bm{G}(\y)\cdot\wn(\d\y) 
\end{align}
as $N\to\infty$, where $\{p_j>0\}=\{\y\in U\colon p_j(\y)>0\}$.
In the case of an infinite index set $\mathcal J$, the series in \eqref{eq:strat_MC_result} converges unconditionally in quadratic mean. 
\end{proposition}

\begin{proof}
First note 
that in order to verify the assertion it is sufficient to 
establish  
for every $\bm{F} \in L^2(\mu;\R^{2d \times \ell})$ 
the convergence in distribution of $\R^{4d}$-valued random variables
\begin{align}\label{eq:MC_convergence_to_prove}
\begin{pmatrix}
\int_U \bm{F}(\y) \cdot \Re\;\!\wn_N(\d \y) \\ 
\int_U  \bm{F}(\y) \cdot \Im\;\! \wn_N(\d\y) 
\end{pmatrix} 
\;\!\xrightarrow{\;\mathrm{d}\;} \;\!
\begin{pmatrix}
\int_{\bigcup_{j}\{p_j>0\}} \bm{F}(\y) \cdot \Re\;\!\wn(\d \y) \\ 
\int_{\bigcup_{j}\{p_j>0\}}  \bm{F}(\y) \cdot \Im\;\! \wn(\d\y) 
\end{pmatrix}
\end{align}
as $N\to \infty$. Indeed, taking into account the structure of complex multiplication, it is clear that \eqref{eq:MC_convergence_to_prove} with $\bm{F} = \begin{psmallmatrix}\Re \bm{G}\\ \Im \bm{G} \end{psmallmatrix}$ and the continuous mapping theorem imply \eqref{eq:strat_MC_result}. We are going to prove \eqref{eq:MC_convergence_to_prove} by applying the central limit theorem. 
To this end, 
observe 
that the $\R^{4d}$-valued random variable on the left-hand side of \eqref{eq:MC_convergence_to_prove} can be written in the form $N^{-1/2} \sum_{n=1}^N \mathbb{X}_n$ with i.i.d.\ $\R^{4d}$-valued random variables $\mathbb{X}_n$, $n \in \N$, defined by
\begin{align}\label{eq:X_n}
\mathbb{X}_n = \sum_{j\in\mathcal J} \frac{1}{\sqrt{p_j(\y_{jn})}}  \begin{pmatrix} \bm{F}(\y_{jn}) \cdot \Re\;\!\wn_{jn} \\ \bm{F}(\y_{jn}) \cdot \Im\;\!\wn_{jn} \end{pmatrix}.
\end{align}
Note that
in the case of an infinite index set $\mathcal J$,
the series \eqref{eq:X_n} converges unconditionally in quadratic mean, i.e., unconditionally in $L^2(\Pr;\R^{4d})$. Indeed, assume w.l.o.g.\ that $\mathcal{J} = \N$. Then the assumptions on $\y_{jn}$, $\wn_{jn}$ imply for all $\eps_j \in \{\pm 1\}$, $J \in \N$, $K \in \N \cup \{\infty\}$ with $J \leq K$ that
\begin{align*}
\E\Biggl[ \biggl\| \sum_{j = J}^K \eps_j \frac{1}{\sqrt{p_j(\y_{jn})}}  \begin{pmatrix} \bm{F}(\y_{jn}) \cdot \Re\;\!\wn_{jn} \\ \bm{F}(\y_{jn}) \cdot \Im\;\!\wn_{jn} \end{pmatrix} \biggr\|^2\Biggr] = \sum_{j = J}^K \int_{\{p_j > 0\}} \bigl\|\bm{F}(\y) \bigr\|^2 \,\mu(\d\y) \leq \sum_{j = J}^K \int_{\Delta_j} \bigl\|\bm{F}(\y) \bigr\|^2 \,\mu(\d\y).
\end{align*}
This 
and the fact that $\sum_{j = 1}^\infty \int_{\Delta_j} \|\bm{F}(\y) \|^2 \,\mu(\d\y) = \int_U \|\bm{F}(\y)\|^2 \,\mu(\d\y) < \infty$ 
yield the asserted unconditional convergence 
in the space $L^2(\Pr;\R^{4d})$. 
Using this result on convergence, the assumptions on $\y_{jn}$, 
$\wn_{jn}$ 
ensure that $\E[\mathbb{X}_n] = \bm{0}$ and that the covariance matrix $\E[ \mathbb{X}_n \otimes \mathbb{X}_n]$ is given by
\begin{align}\label{eq:MC_covariance}
\begin{pmatrix}
\frac12 \int_{\bigcup_j\{p_j>0\}} \bm{F}(\y) \cdot \bm{F}(\y)^\top \mu(\d \y) & \bm{0} \\
\bm{0} & \frac12 \int_{\bigcup_j\{p_j>0\}} \bm{F}(\y) \cdot \bm{F}(\y)^\top \mu(\d \y)
\end{pmatrix} \in \R^{4d \times 4d}.
\end{align}
Moreover, 
observe 
that the isometric 
identities 
\eqref{eq:isometric_property} 
for  
white noise  
integrals 
imply that the matrix in \eqref{eq:MC_covariance} coincides with the covariance matrix of the $\R^{4d}$-valued centered Gaussian random variable on the right-hand side of \eqref{eq:MC_convergence_to_prove};
compare \cite[Lemma A.2]{AKLMW24}.
An application of the multidimensional central limit theorem thus 
yields 
the convergence in \eqref{eq:MC_convergence_to_prove}.
\end{proof}

As shown in Lemma~\ref{lem:covariance} below, the 
integrals \wrt the discrete random measures $\wn_N$ 
in Proposition~\ref{prop:stratMC} 
satisfy isometric 
identities 
and covariance 
formulas  
that are 
analogous 
to those 
for 
stochastic integrals \wrt 
a
Gaussian white noise $\wn$; 
compare \eqref{eq:isometric_property} and \cite[Appendix~A]{AKLMW24}. 
Since we are mainly interested in the real parts of the complex-valued stochastic integrals, 
we 
restrict the formulation of the results 
accordingly.

\begin{lemma}[Isometric property, covariance formula] \label{lem:covariance}
Assume the setting 
given in
Proposition~\ref{prop:stratMC}, let $\bm{G},\bm{H}\in L^2(\mu;\C^{d\times\ell})$, 
$N\in\N$,  
and assume for every $j\in\mathcal J$ that 
$\mu\bigl(\bigl\{\y\in \Delta_j\colon p_j(\y)=0 \text{ and }$ $ (\bm{G}(\y),\bm{H}(\y))\neq(\bm0,\bm0)\bigr\}\bigr)=0$.
Then, 
we have  
\begin{align}\label{eq:isometric_property_1}
\E\biggl[ \Bigr\|\Re\!\int_U \bm{G}(\y)\cdot\wn_N(\d\y) \Bigl\|^2 \biggr] 
= 
\frac12\int_U \bigl\| \bm{G}(\y)\bigr\|^2 \mu(\d\y)
\end{align}
and 
\begin{align}\label{eq:isometric_property_2}
\E\biggl[ \Re\!\int_U \bm{G}(\y)\cdot\wn_N(\d\y) \,\otimes\, \Re\!\int_U \bm{H}(\y)\cdot\wn_N(\d\y)\biggr] 
= 
\frac12\,\Re \! \int_U \bm{G}(\y)\cdot \overline{\bm{H}(\y)}^\top \mu(\d\y).
\end{align}
\end{lemma} 

\begin{proof}
Observe that the assumptions regarding the distributions and 
independence 
properties  
of the random variables $\y_{jn}$, $\Re\;\!\wn_{jn}$ 
imply for every real-valued integrand $\bm{F} \in L^2(\mu;\R^{d\times\ell})$ that
\begin{equation}\label{eq:isometric_property_3}
\begin{aligned}
\E\biggl[ \Bigr\| \int_U \bm{F}(\y)\cdot\Re\;\!\wn_N(\d\y) \Bigl\|^2 \biggr] 
& = 
\sum_{j,k \in \mathcal J} \frac{1}{N} \sum_{n,m=1}^N \E \biggl[ \frac{\Re\;\!\wn_{jn} \cdot \bm{F}(\y_{jn})^\top \cdot \bm{F}(\y_{km}) \cdot \Re\;\!\wn_{km}}{\sqrt{p_j(\y_{jn})\;\! p_k(\y_{km}) }} \biggr]
\\ & = \sum_{j\in \mathcal J} \frac{1}{N} \sum_{n=1}^N \E \biggl[ \frac{1}{p_j(\y_{jn})} \bigl\| \bm{F}(\y_{jn}) \cdot \Re\;\!\wn_{jn}\bigr\|^2 \biggr] 
\\ & = 
\sum_{j\in \mathcal J} \frac{1}{2} \int_{\{p_j > 0\}} \bigl\| \bm{F}(\y) \bigr\|^2 \mu(\d\y) 
= \frac{1}{2} \int_{\bigcup_{j}\{p_j>0\}} \bigl\| \bm{F}(\y) \bigr\|^2 \mu(\d\y),
\end{aligned}
\end{equation}
where we have used the fact that the 
expected values 
appearing in the penultimate line do not depend on $n$ and are equal to 
$\int_{\{p_j > 0\}} \| \bm{F}(\y) \|^2 \mu(\d\y)$. 
Moreover, note that \eqref{eq:isometric_property_3} remains true if the real parts of $\wn_N$, $\wn_{jn}$, and $\wn_{km}$ are replaced by the corresponding imaginary parts. 
Applying these identities with 
$\Re\;\! \bm{G}$ and $\Im \;\!\bm{G}$ in place of $\bm{F}$ 
and using the uncorrelatedness of $\Re\;\! \wn_N$ and $\Im \;\!\wn_N$ 
ensures that  
\begin{align*}
\E\biggl[ \Bigr\|\Re\!\int_U \bm{G}(\y)\cdot\wn_N(\d\y) \Bigl\|^2 \biggr] & = \E\biggl[ \Bigr\| \int_U \Re\;\!\bm{G}(\y)\cdot \Re\;\! \wn_N(\d\y) \Bigl\|^2 \biggr] + \E\biggl[ \Bigr\| \int_U \Im\;\!\bm{G}(\y)\cdot \Im\;\! \wn_N(\d\y) \Bigl\|^2 \biggr]
\\ & = \frac12\int_{\bigcup_{j}\{p_j>0\}} \bigl\| \bm{G}(\y)\bigr\|^2 \mu(\d\y).
\end{align*} 
This
and the fact that 
$\mu(\{\y\in \Delta_j\colon p_j(\y)=0 \text{ and }\bm{G}(\y)\neq\bm0\})=0$ 
establish 
the isometric identity 
\eqref{eq:isometric_property_1}.
The covariance formula 
\eqref{eq:isometric_property_2} 
follows directly from 
\eqref{eq:isometric_property_1} 
and the 
polarization 
identity 
$\Re(u\overline{v}) = (|u+v|^2 - |u|^2 - |v|^2)/2$ 
for 
$u,v \in \C$.
\end{proof}

As a consequence of Proposition~\ref{prop:stratMC} and Lemma~\ref{lem:covariance},  
we obtain the following result on stratified Monte Carlo approximations of random fields given in terms of stochastic integrals.

\begin{corollary}[Monte Carlo integration for random fields]\label{cor:stratMC_randomfields}
Assume the setting of Proposition~\ref{prop:stratMC}, 
let 
$\bm{G}\colon\R^{n_0}\times U\to\C^{d\times\ell}$ 
satisfy 
for every 
$\x\in\R^{n_0}$ 
that $\bm{G}(\x,\bdot)\in L^2(\mu;\C^{d\times\ell})$,   
and assume for every $j\in\mathcal J$, $\x\in\R^{n_0}$ that
$\mu(\{\y\in \Delta_j\colon p_j(\y)=0 \text{ and } \bm{G}(\x,\y)\neq\bm0\})=0$.
Let the $\R^d$-valued random fields 
$\v=(\v(\x))_{\x\in\R^{n_0}}$ and $\v_N=(\v_N(\x))_{\x\in\R^{n_0}}$, 
$N\in\N$, be such that for every 
$x\in\R^{n_0}$ 
it holds $\Pr$-almost surely that
\begin{align*}
\v(\x) = \Re \! \int_U \bm{G}(\x,\y)\cdot\wn(\d\y)
\end{align*}
and 
\begin{align*}
\v_N(\x) = \Re \! \int_U \bm{G}(\x,\y)\cdot\wn_N(\d\y)
= 
\sum_{j\in\mathcal J}\frac1{\sqrt{N}} \sum_{n=1}^N \frac1{\sqrt{p_j(\y_{jn})}} \;\! \Re\bigl[ \bm{G}(\x,\y_{jn})\cdot\wn_{jn} \bigr].
\end{align*}
Then, the finite-dimensional distributions of $\v_N$ converge 
weakly 
to the respective 
finite-dimensional distributions of $\v$ as $N\to\infty$, 
i.e., for any choice of points $\x_1$, $\x_2$, \ldots, $\x_k\in\R^{n_0}$, $k\in\N$, it holds that 
\begin{align*}
\bigl(\v_N(\x_1), \v_N(\x_2),\ldots,\v_N(\x_k)\bigr) 
\;\!\xrightarrow{\;\mathrm{d}\;} \;\!
\bigl(\v(\x_1), \v(\x_2),\ldots,\v(\x_k)\bigr). 
\end{align*} 
Moreover, 
for every $N\in\N$ the covariance structure of the random field $\v_N$ is identical 
to the covariance structure of $\v$. 
\end{corollary} 

\begin{proof}
The first statement follows 
readily 
from 
an application of 
Proposition~\ref{prop:stratMC} 
using 
the 
integrand
$(\bm{G}(\x_1,\bdot),\ldots,\bm{G}(\x_k,\bdot))$  
considered as an element in $L^2(\mu;\C^{kd\times \ell})$. 
The second statement 
is a 
consequence 
of the covariance identities in 
Lemma \ref{lem:covariance}
and \cite[Lemma A.2]{AKLMW24}.
\end{proof}

 
\printbibliography

@article{Kur93,
	author = {Kurbanmuradov, O.},
	journal = {Bulletin of the Novosibirsk Computing Center, Numerical Analysis},
	pages = {19--25},
	title = {Weak convergence of randomized spectral methods of {G}aussian random vector fields},
	volume = {4},
	year = {1993}
	}

@article{SD96,
	author = {Shinozuka, M. and Deodatis, G.},
	doi = {10.1115/1.3101883},
	journal = {Applied Mechanics Reviews},
	number = {1},
	pages = {29--53},
	title = {Simulation of multi-dimensional {G}aussian stochastic fields by spectral representation},
	volume = {49},
	year = {1996}
}

@article{Shi72,
	author = {Shinozuka, M.},
	doi = {10.1016/0045-7949(72)90043-0},
	journal = {Computers \& Structures},
	number = {5},
	pages = {855--874},
	title = {Monte {C}arlo solution of structural dynamics},
	volume = {2},
	year = {1972}
}

@article{DS25,
	author = {Deodatis, G. and Shields, M.},
	doi = {10.1016/j.ress.2024.110522},
	journal = {Reliability Engineering \& System Safety},
	pages = {110522},
	title = {The Spectral Representation Method: {A} framework for simulation of stochastic processes, fields, and waves},
	volume = {254},
	year = {2025}
}

@book{Sab13,
	author = {Sabelfeld, K.},
	doi = {10.1515/9783110296815},
	publisher = {De Gruyter},
	title = {Random Fields and Stochastic Lagrangian Methods - Analysis and Applications in Turbulence and Porous Media},
	year = {2013}
}

@article{VS21,
	author = {Vandanapu, L. and Shields, M. D.},
	doi = {10.1016/j.probengmech.2021.103128},
	journal = {Probabilistic Engineering Mechanics},
	pages = {103128},
	title = {3rd-order Spectral Representation Method: {S}imulation of multi-dimensional random fields and ergodic multi-variate random processes with fast {F}ourier transform implementation},
	volume = {64},
	year = {2021}
}

@article{Zwi+24,
	author = {Zwijsen, K. and {van den Bos}, N. and Frederix, E. M. A. and Roelofs, F. and {van Zuijlen}, A. H.},
	doi = {10.1016/j.nucengdes.2024.113316},
	journal = {Nuclear Engineering and Design},
	pages = {113316},
	title = {Development of an anisotropic pressure fluctuation model for the prediction of turbulence-induced vibrations of fuel rods},
	volume = {425},
	year = {2024}
}

@article{KCX13,
	author = {Kim, Y. and Castro, I. P. and Xie, Z.-T.},
	doi = {10.1016/j.compfluid.2013.06.001},
	journal = {Computers \& Fluids},
	pages = {56--68},
	title = {Divergence-free turbulence inflow conditions for large-eddy simulations with incompressible flow solvers},
	volume = {84},
	year = {2013}
}

@article{KSJ03,
	author = {Klein, M. and Sadiki, A. and Janicka, J.},
	doi = {10.1016/S0021-9991(03)00090-1},
	journal = {Journal of Computational Physics},
	number = {2},
	pages = {652--665},
	title = {A digital filter based generation of inflow data for spatially developing direct numerical or large eddy simulations},
	volume = {186},
	year = {2003}
}

@inproceedings{Ewe07,
	author = {Ewert, R.},
	doi = {10.2514/6.2007-3506},
	booktitle = {13th AIAA/CEAS Aeroacoustics Conference (28th AIAA Aeroacoustics Conference)},
	title = {{RPM} -- the fast {R}andom {P}article-{M}esh method to realize unsteady turbulent sound sources and velocity fields for {CAA} applications},
	year = {2007}
}

@article{Pol+20,
	author = {Polacsek, C. and Cader, A. and Buszyk, M. and Barrier, R. and Gea-Aguilera, F. and Posson, H.},
	doi = {10.1063/5.0020190},
	journal = {Physics of Fluids},
	number = {10},
	pages = {107107},
	title = {Aeroacoustic design and broadband noise predictions of a fan stage with serrated outlet guide vanes},
	volume = {32},
	year = {2020}
}

@article{Wie+19,
	author = {Wieland, M. and Arne, W. and Marheineke, N. and Wegener, R.},
	doi = {10.1016/j.apm.2019.06.023},
	journal = {Applied Mathematical Modelling},
	pages = {558--577},
	title = {Melt-blowing of viscoelastic jets in turbulent airflows: {S}tochastic modeling and simulation},
	volume = {76},
	year = {2019}
}

@article{XYD22,
	author = {Xue, X. and Yao, H.-D. and Davidson, L.},
	doi = {10.1063/5.0090641},
	journal = {Physics of Fluids},
	number = {5},
	pages = {055118},
	title = {Synthetic turbulence generator for lattice {B}oltzmann method at the interface between {RANS} and {LES}},
	volume = {34},
	year = {2022}
}

@article{DL18,
	author = {Du, Y. and Lin, G.},
	doi = {10.1098/rspa.2018.0093},
	journal = {Proceedings of the Royal Society A: Mathematical, Physical and Engineering Sciences},
	number = {2217},
	pages = {1--19},
	title = {Turbulence generation from a stochastic wavelet model},
	volume = {474},
	year = {2018}
}

@inproceedings{Ewe16,
	author = {Ewert, R.},
	doi = {10.2514/6.2016-2965},
	booktitle = {22nd AIAA/CEAS Aeroacoustics Conference},
	title = {Canonical stochastic realization of turbulent sound sources via forced linear advection-diffusion-dissipation equation},
	year = {2016}
}

@inproceedings{BED03,
	author = {Billson, M. and Eriksson, L.-E. and Davidson, L.},
	doi = {10.2514/6.2003-3282},
	booktitle = {9th AIAA/CEAS Aeroacoustics Conference and Exhibit},
	title = {Jet noise prediction using stochastic turbulence modeling},
	year = {2003}
}

@article{BBLC94,
	author = {Bechara, W. and Bailly, C. and Lafon, P. and Candel, S. M.},
	doi = {10.2514/3.12008},
	journal = {AIAA Journal},
	number = {3},
	pages = {455--463},
	title = {Stochastic approach to noise modeling for free turbulent flows},
	volume = {32},
	year = {1994}
}

@book{Dav15,
	author = {Davidson, P.},
	publisher = {Oxford University Press},
	title = {Turbulence: {A}n Introduction for Scientists and Engineers},
	year = {2015}
}

@misc{AKLMW24,
	author = {Antoni, M. and K\"{u}rpick, Q. and Lindner, F. and Marheineke, N. and Wegener, R.},
	doi = {10.48550/arXiv.2311.09893},
	howpublished = {preprint arXiv:2311.09893v3},
	title = {Random field reconstruction of inhomogeneous turbulence. {P}art {I}: {M}odeling and analysis},
	year = {2026}
}

@article{Kur95,
	author = {Kurbanmuradov, O.},
	doi = {10.1515/rnam.1995.10.4.311},
	journal = {Russian Journal of Numerical Analysis and Mathematical Modelling},
	number = {4},
	pages = {311--323},
	title = {Convergence of numerical models for the {G}aussian fields},
	volume = {10},
	year = {1995}
}

@inproceedings{PS95,
	author = {Poirion, F. and Soize, C.},
	doi = {10.1007/3-540-60214-3_50},
	booktitle = {Probabilistic Methods in Applied Physics},
	pages = {17--53},
	publisher = {Springer},
	title = {Numerical methods and mathematical aspects for simulation of homogeneous and non homogeneous {G}aussian vector fields},
	year = {1995}
}

@article{BK95,
	author = {Buglanova, N. A. and Kurbanmuradov, O. A.},
	doi = {10.1515/mcma.1995.1.3.173},
	journal = {Monte Carlo Methods and Applications},
	number = {3},
	pages = {173--201},
	title = {Convergence of the randomized spectral models of homogeneous {G}aussian random fields},
	volume = {1},
	year = {1995}
}

@book{Pri01,
	author = {Prigarin, S. M.},
	publisher = {De Gruyter},
	title = {Spectral Models of Random Fields in Monte Carlo Methods},
	year = {2001}
}

@book{MNR12,
	author = {M\"uller-Gronbach, T. and Novak, E. and Ritter, K.},
	doi = {10.1007/978-3-540-89141-3},
	publisher = {Springer},
	title = {Monte {C}arlo-Algorithmen},
	year = {2012}
}

@book{RK17,
	author = {Rubinstein, R. Y. and Kroese, D. P.},
	doi = {10.1002/9781118631980},
	publisher = {John Wiley \& Sons},
	title = {Simulation and the {M}onte {C}arlo Method},
	year = {2017}
}

@article{ADDK22,
	author = {Alexandrov, A. V. and Dorodnicyn, L. W. and Duben, A. P. and Kolyukhin, D. R.},
	doi = {10.1134/S2070048222010033},
	journal = {Mathematical Models and Computer Simulations},
	number = {1},
	pages = {92--98},
	title = {Generation of anisotropic turbulent velocity fields based on a randomized spectral method},
	volume = {14},
	year = {2022}
}

@article{Kra70,
	author = {Kraichnan, R. H.},
	doi = {10.1063/1.1692799},
	journal = {Physics of Fluids},
	number = {1},
	pages = {22--31},
	title = {Diffusion by a random velocity field},
	volume = {13},
	year = {1970}
}

@article{SSC01,
	author = {Smirnov, A. and Shi, S. and Celik, I.},
	doi = {10.1115/1.1369598},
	journal = {Journal of Fluids Engineering},
	number = {2},
	pages = {359--371},
	title = {Random flow generation technique for large eddy simulations and particle-dynamics modeling},
	volume = {123},
	year = {2001}
}

@article{YB14,
	author = {Yu, R. and Bai, X.},
	doi = {10.1016/j.jcp.2013.08.055},
	journal = {Journal of Computational Physics},
	pages = {234--253},
	title = {A fully divergence-free method for generation of inhomogeneous and anisotropic turbulence with large spatial variation},
	volume = {256},
	year = {2014}
}

@article{SSST14,
	author = {Shur, M. L. and Spalart, P. R. and Strelets, M. K. and Travin, A. K.},
	doi = {10.1007/s10494-014-9534-8},
	journal = {Flow, Turbulence and Combustion},
	number = {1},
	pages = {63--92},
	title = {Synthetic turbulence generators for {RANS}-{LES} interfaces in zonal simulations of aerodynamic and aeroacoustic problems},
	volume = {93},
	year = {2014}
}

@article{GJYZ23,
	author = {Guo, H. and Jiang, P. and Ye, L. and Zhu, Y.},
	doi = {10.1017/jfm.2023.548},
	journal = {Journal of Fluid Mechanics},
	pages = {A2},
	title = {An efficient and low-divergence method for generating inhomogeneous and anisotropic turbulence with arbitrary spectra},
	volume = {970},
	year = {2023}
}

@article{Mann98,
	author = {Mann, J.},
	doi = {10.1016/S0266-8920(97)00036-2},
	journal = {Probabilistic Engineering Mechanics},
	number = {4},
	pages = {269--282},
	title = {Wind field simulation},
	volume = {13},
	year = {1998}
}

@book{Pope00,
	author = {Pope, S. B.},
	publisher = {Cambridge University Press},
	title = {Turbulent Flows},
	year = {2000}
}

@book{MY75,
	author = {Monin, A. S. and Yaglom, A. M.},
	publisher = {Dover Publications},
	title = {Statistical Fluid Mechanics. {V}olume {II}: {M}echanics of Turbulence},
	year = {2007}
}

@book{Bre14,
	author = {Br\'{e}maud, P.},
	doi = {10.1007/978-3-319-09590-5},
	publisher = {Springer},
	title = {Fourier Analysis and Stochastic Processes},
	year = {2014}
}

@article{HMRW13,
	author = {H\"{u}bsch, F. and Marheineke, N. and Ritter, K. and Wegener, R.},
	doi = {10.1007/s10955-013-0715-y},
	journal = {Journal of Statistical Physics},
	number = {6},
	pages = {1115--1137},
	title = {Random field sampling for a simplified model of melt-blowing considering turbulent velocity fluctuations},
	volume = {150},
	year = {2013}
}

@article{MW11,
	author = {Marheineke, N. and Wegener, R.},
	doi = {10.1016/j.ijmultiphaseflow.2010.10.001},
	journal = {International Journal of Multiphase Flow},
	pages = {136--148},
	title = {Modeling and application of a stochastic drag for fiber dynamics in turbulent flows},
	volume = {37},
	year = {2011}
}

@book{GS74,
	author = {Gikhman, I. I. and Skorokhod, A. V.},
	doi = {10.1007/978-3-642-61943-4},
	publisher = {Springer},
	title = {The Theory of Stochastic Processes {I}},
	year = {2004}
}

@book{LPS14,
	author = {Lord, G. J. and Powell, C. E. and Shardlow, T.},
	doi = {10.1017/CBO9781139017329},
	publisher = {Cambridge University Press},
	title = {An Introduction to Computational Stochastic {PDE}s},
	year = {2014}
}

@article{KS06,
	author = {Kurbanmuradov, O. and Sabelfeld, K.},
	doi = {10.1515/156939606779329080},
	journal = {Monte Carlo Methods and Applications},
	number = {5-6},
	pages = {395--445},
	title = {Stochastic spectral and {F}ourier-wavelet methods for vector {G}aussian random fields},
	volume = {12},
	year = {2006}
}

@article{KKS07,
	author = {Kramer, P. R. and Kurbanmuradov, O. and Sabelfeld, K.},
	doi = {10.1016/j.jcp.2007.05.002},
	journal = {Journal of Computational Physics},
	number = {1},
	pages = {897--924},
	title = {Comparative analysis of multiscale {G}aussian random field simulation algorithms},
	volume = {226},
	year = {2007}
}

@article{KS08,
	author = {Kurbanmuradov, O. and Sabelfeld, K.},
	doi = {10.1137/070699408},
	journal = {SIAM Journal on Numerical Analysis},
	number = {6},
	pages = {3084--3112},
	title = {Convergence of {F}ourier-wavelet models for {G}aussian random processes},
	volume = {46},
	year = {2008}
}

@article{EM95,
	author = {Elliott, F. W. and Majda, A. J.},
	doi = {10.1006/jcph.1995.1052},
	journal = {Journal of Computational Physics},
	number = {1},
	pages = {146--162},
	title = {A new algorithm with plane waves and wavelets for random velocity fields with many spatial scales},
	volume = {117},
	year = {1995}
}

@article{ADDK20,
	author = {Aleksandrov, A. V. and Dorodnitsyn, L. V. and Duben, A. P. and Kolyukhin, D. R.},
	doi = {10.1007/s10598-020-09493-9},
	journal = {Computational Mathematics and Modeling},
	number = {3},
	pages = {308--319},
	title = {Generation of nonhomogeneous turbulent velocity fields by modified randomized spectral method},
	volume = {31},
	year = {2020}
}

@article{LCS07,
	author = {Liang, J. and Chaudhuri, S. R. and Shinozuka, M.},
	doi = {/10.1061/(ASCE)0733-9399(2007)133:6(616)},
	journal = {Journal of Engineering Mechanics},
	number = {6},
	pages = {616--627},
	title = {Simulation of nonstationary stochastic processes by spectral representation},
	volume = {133},
	year = {2007}
}

@article{SJ72,
	author = {Shinozuka, M. and Jan, C.-M.},
	doi = {10.1016/0022-460X(72)90600-1},
	journal = {Journal of Sound and Vibration},
	number = {1},
	pages = {111--128},
	title = {Digital simulation of random processes and its applications},
	volume = {25},
	year = {1972}
}

@article{Shi71,
	author = {Shinozuka, M.},
	doi = {10.1121/1.1912338},
	journal = {The Journal of the Acoustical Society of America},
	number = {1B},
	pages = {357--368},
	title = {Simulation of multivariate and multidimensional random processes},
	volume = {49},
	year = {1971}
}

@article{KSK13,
	author = {Kurbanmuradov, O. and Sabelfeld, K. and Kramer, P. R.},
	doi = {10.1016/j.jcp.2013.03.021},
	journal = {Journal of Computational Physics},
	pages = {218--234},
	title = {Randomized spectral and {F}ourier-wavelet methods for multidimensional {G}aussian random vector fields},
	volume = {245},
	year = {2013}
}

\end{document}